\documentclass[11pt,preprint]{amsart}

\usepackage{amsfonts}
\RequirePackage{natbib}
\usepackage{graphicx}
\usepackage{color}
\usepackage{amssymb}
\usepackage{amsfonts}
\usepackage{amsthm}
\usepackage[mathcal]{eucal}
\usepackage{ipa}

\usepackage[justification=centering,font=footnotesize]{caption}
\usepackage[labelformat=simple]{subfig}

\newtheorem{theoreme}{Theorem}[section]
\newtheorem{lemma}[theoreme]{Lemma}
\newtheorem{proposition}[theoreme]{Proposition}

\newtheorem{definition}[theoreme]{Definition}
\newtheorem{remark}[theoreme]{Remark}
\newtheorem{corollary}[theoreme]{Corollary}

\newtheorem{example}[theoreme]{Example}
\newtheorem{bio}[theoreme]{Biological Interpretation}

\def\R{\mathbb R}
\def\N{\mathbb N}
\def\Z{\mathbb Z}

\newcommand{\begitem}{\begin{itemize}}
\newcommand{\finit}{\end{itemize}}

\newcommand\restr[2]{{% we make the whole thing an ordinary symbol
  \left.\kern-\nulldelimiterspace % automatically resize the bar with \right
  #1 % the function
  \vphantom{\big|} % pretend it's a little taller at normal size
  \right|_{#2} % this is the delimiter
  }}

\newcommand{\bone}{\mathbf{1}}

\newcommand{\diag}{\mathrm{diag}}
\newcommand{\bS}{S}

\newcommand{\bzeta}{\boldsymbol{\zeta}}

\newcommand{\PPinvphi}{\Inv_{\mathbb{Q}}(\Phi)}

\newcommand{\PPinvtheta}{\Inv_{\mathbb{Q}}(\Theta)}
\newcommand{\PPinv}{\Inv(\theta)}
\newcommand{\PP}{\mathcal{P}}

\newcommand{\E}{\mathbb{E}}
\newcommand{\Prob}{\mathbb{P}}
\newcommand{\ProbQ}{\mathbb{Q}}

\newcommand{\eps}{\varepsilon}
\newcommand{\rr}{r}

\DeclareMathOperator{\interior}{int}

\DeclareMathOperator{\Inv}{Inv}

\begin{document}

\title[Stochastic persistence of interacting structured populations]{Persistence in fluctuating environments for interacting structured populations }

\author[G. Roth]{Gregory Roth}
\email{groth@ucdavis.edu}
\author[S.J. Schreiber]{Sebastian J. Schreiber}
\email{sschreiber@ucdavis.edu}
\address{Department of Evolution and Ecology, One Shields Avenue, University of California, Davis, CA 95616 USA}

\begin{abstract}
Individuals within any species exhibit differences in size, developmental state, or spatial location. These differences coupled with environmental fluctuations in demographic rates can have subtle effects on population persistence and species coexistence. To understand these effects, we provide a general theory for coexistence of structured, interacting species living in a stochastic environment.  The theory is applicable to nonlinear, multi species matrix models with stochastically varying parameters. The theory relies on long-term growth rates of species corresponding to the  dominant Lyapunov exponents of random   matrix products. Our coexistence criterion requires that a convex combination of these long-term growth rates is positive with probability one whenever one or more species are at low density.  When this condition holds, the community is stochastically persistent: the fraction of time that a species density goes below $\delta>0$ approaches zero as $\delta$ approaches zero.  Applications to predator-prey interactions in an autocorrelated environment, a stochastic LPA model, and spatial lottery models are provided. These applications demonstrate that positive autocorrelations in temporal fluctuations can disrupt predator-prey coexistence, fluctuations in log-fecundity can facilitate persistence in structured populations, and long-lived, relatively sedentary competing populations are likely to coexist in spatially and temporally heterogenous environments. 

\end{abstract}

\maketitle
\section{Introduction}
All populations are structured and experience environmental fluctuations. Population structure may arise to individual differences in age, size, and spatial location~\citep{metz-diekmann-86,caswell-01,holyoak-etal-05}. Temporal fluctuations in environmental factors such light, precipitation, and temperature occur in all natural marine, freshwater and terrestrial systems. Since these environmental factors can influence survival, growth, and reproduction, environmental fluctuations result in demographic fluctuations that may influence species persistence and the composition of ecological communities~\citep{tuljapurkar-90,chesson-00b,kuang-chesson-09}. Here we present, for the first time, a general approach to studying coexistence of structured populations in fluctuating environments.

For species interacting in an ecosystem, a fundamental question is what are the minimal conditions to ensure the long-term persistence of all species. Historically, theoretical ecologists characterize persistence by the existence of an asymptotic equilibrium in which the proportion of each population is strictly positive~\citep{may-75,roughgarden-79}. More recently, coexistence was equated  with the existence of an attractor bounded away from extinction~\citep{hastings-88}, a definition that ensures populations will persist despite small, random perturbations of the populations~\citep{jtb-06,dcds-07}. However, ``environmental perturbations are often vigourous shake-ups, rather than gentle stirrings''~\citep{jansen-sigmund-98}. To account for large, but rare, perturbations,  the concept of permanence, or uniform persistence, was introduced in late 1970s~\citep{freedman-waltman-77,schuster-etal-79}. Uniform persistence requires that asymptotically species densities remain uniformly bounded away from extinction. In addition, permanence requires that the system is dissipative i.e. asymptotically species densities remain uniformly bounded from above.  Various mathematical approaches exist for verifying permanence~\citep{hutson-schmitt-92,smith-thieme-11} including topological characterizations with respect to chain recurrence~\citep{butler-waltman-86,hofbauer-so-89}, average Lyapunov functions~\citep{hofbauer-81,hutson-84,garay-hofbauer-03}, and measure theoretic approaches~\citep{jde-00,jde-10}. The latter two approaches involve the long-term, per-capita growth rates of species when rare. For discrete-time, unstructured models of the form $x_{t+1}^i= f_i(x_t)x_t^i$ where $x_t=(x_t^1,\dots,x_t^n)$ is the vector of population densities at time $t$, the long-term growth rate of species $i$ with initial community state $x_0=x$ equals \[
r_i (x) = \limsup_{t\to\infty}\frac{1}{t}\sum_{s=0}^{t-1} \log f_i (x_s).\] 
\citet{garay-hofbauer-03} showed, under appropriate assumptions, that the system is permanent provided there exist positive weights $p_1,\dots,p_n$ associated with each species such that $\sum_i p_i r_i (x)>0$ for any initial condition $x$ with one or more missing species (i.e. $\prod_i x^i=0$). Intuitively, the community persists if on average the community increases when rare. 

The permanence criterion for unstructured populations also extends to structured populations. However, in this case, the long-term growth rate is more complicated. Consider, for example, when both time and the structuring variables are discrete; the population dynamics are given by  $x_{t+1}^i=A_i(x_t)x_t^i$ where $x_t^i$ is a vector corresponding to the densities of the stages of species $i$, $x_t=(x_t^1,\dots,x_t^n)$, and $A_i(x)$ are non-negative matrices. Then the long term growth rate $r_i(x)$ of species $i$ corresponds to the dominant Lyapunov exponent associated with the matrices $A_i(x)$ along the population trajectory: \[
r_i(x) = \limsup_{t\to\infty} \frac{1}{t} \log \| A_i(x_{t-1})\dots A_i(x_0)\|.
\] 
At the extinction state $x=0$, the long-term growth rate $r_i(0)$ simply corresponds to the $\log$  of the  largest eigenvalue of $A_i(0)$. For structured single-species models, \citet{cushing-98,kon-etal-04} proved that $r_1(0)>0$ implies permanence. For structured, continuous-time, multiple species models, $r_i(x)$ can be defined in an analogous manner to the discrete-time case using the fundamental matrix of the variational equation. \citet{jde-10} showed, under appropriate assumptions, that $\sum_i p_i r_i(x)>0$ for all $x$ in the extinction set is sufficient for permanence. For discrete-time structured models, however, there exists no general proof of this fact (see, however, \citet{salceanu-smith-09,salceanu-smith-09lyapunov,salceanu-smith-10}). When both time and the structuring variables are continuous, the models become infinite dimensional and may be formulated as partial differential equations or functional differential equations. Much work has been done is this direction~\citep{hutson-moran-87,zhao-hutson-94,thieme-09,thieme-11,magal-mccluskey-webb-10,xu-zhao-03,jin-zhao-09}. In particular, for reaction-diffusion equations, the long-term growth rates correspond to growth rates of semi-groups of linear operators and, $\sum_i p_i r_i(x)>0$ for all $x$ in the extinction set also ensures permanence for these models~\citep{hutson-moran-87,zhao-hutson-94,cantrell-cosner-03}.  
    
Environmental stochasticity can be a potent force for disrupting population persistence yet maintaining biodiversity. Classical stochastic demography theory for stochastic matrix models $x_{t+1}=A(t)x_t$ shows that temporally uncorrelated fluctuations in the projection matrices $A(t)$ reduce the long-term growth rates of populations when rare~\citep{tuljapurkar-90,boyce-etal-06}.  Hence, increases in the magnitude of these uncorrelated fluctuations can shift  populations from persisting to asymptotic extinction. Under suitable conditions, the long-term growth rate for these models is given by the limit $r=\lim_{t\to\infty} \frac{1}{t} \ln \| A(t)\dots A(1)\|$ with probability one. When $r>0$, the population grows exponentially with probability one for these density-independent models. When $r<0$, the population declines exponentially with probability one. \citet{hardin-etal-88a} and \citet{tpb-09} proved that these conclusions extend to models with compensating density-dependence. However, instead of growing without bound when $r>0$,  the populations converge to a positive stationary distribution  with probability one. These results, however, do not apply to models with over-compensating density-dependence or, more generally, non-monotonic responses of demography to density. 
 
Environmental stochasticity can promote diversity through the storage effect~\citep{chesson-warner-81,chesson-82} in which asynchronous fluctuations of favorable conditions can allow long-lived species competing for space to coexist.  The theory for coexistence in stochastic environments has focused on stochastic difference equations  of the form $x_{t+1}^i=x_t^i f_i (\xi_{t+1},x_t)$ where $\xi_1,\xi_2,\dots$ is a sequence of independent, identically distributed random variables (for a review see ~\citep{jdea-11}). \citet{jmb-11} prove that coexistence, in a suitable sense, occurs provided that $\sum_i p_i r_i(x)>0$ with probability for all $x$ in the extinction set. Similar to the deterministic case, the long-term growth rate of species $i$ equals $r_i (x) = \limsup_{t\to\infty}\frac{1}{t}\sum_{s=0}^{t-1} \log f_i (x_s)$.  Here, stochastic coexistence implies that each species spends an arbitrarily small fraction of time near arbitrarily small densities.

Here, we develop persistence theory for models simultaneously accounting for species interactions, population structure, and environmental fluctuations.  Our main result implies that the ``community increases when rare'' persistence criterion also applies to these models. Our model, assumptions, and a definition of stochastic persistence are presented in Section 2. Except for a compactness assumption, our assumptions are quite minimal allowing for overcompensating density dependence and correlated environmental fluctuations. Long-term growth rates for these models and our main theorem are stated in Section 3. We apply our results to stochastic models of predator-prey interactions, stage-structured beetle dynamics, and competition in spatial heterogenous environments. The stochastic models for predator-prey interactions are presented in Section 4 and examine to what extent ``colored'' environmental fluctuations facilitate predator-prey coexistence. In Section 5, we develop precise criteria for persistence and exclusion for structured single species models and apply these results to the classic stochastic model of larvae-pupae-adult dynamics of flour beetles~\citep{costantino-etal-95,dennis-etal-95,costantino-etal-97,henson-cushing-97} and metapopulation dynamics~\citep{harrison-quinn-89,gyllenberg-etal-96,metz-gyllenberg-01,roy-etal-05,hastings-botsford-06,prsb-10}. We show, contrary to initial expectations, that multiplicative noise with logarithmic means of zero can facilitate persistence. In Section 6, we examine  spatial-explicit lottery models~\citep{chesson-85,chesson-00a,chesson-00b} to illustrate how spatial and temporal heterogeneity, collectively, mediate coexistence for transitive and intransitive competitive communities. Proofs of most results are presented in Section \ref{sec-proof}.

%*******************************************************************************
\section{Model and assumptions}
 
%*******************************************************************************
We study the dynamics of $m$ interacting populations in a random environment. Each individual in population $i$ can be in one of $n_i$ individual states such as their age, size, or location.  Let $X_t^i = (X_t^{i1}, \dots, X_t^{in_i})$ denote the row vector of populations abundances of individuals in different states for population $i$ at time $t\in \N$. $X_t^i$ lies in the non-negative cone $\R^{n_i}_+$. The \emph{population state} is the row vector  $X_t=(X_t^1, \dots,  X_t^m)$ that lies in the non-negative cone $\R^n_+$ where $n=\sum_{i=1}^mn_i$. 
To account for environment fluctuations, we consider a sequence of random variables, $\xi_1,\xi_2,\dots,\xi_t,\dots$ where $\xi_t$ represents the state of the environment at time $t$. 

To define the population dynamics, we consider projection matrices for each population that depend on the population state and the environmental state. More precisely, for each $i$, let $A_i(\xi,X)$ be a non-negative, $n_i\times n_i$ matrix whose $j$--$k$-th entry corresponds to the contribution of individuals in state $j$ to individuals in state $k$ e.g. individuals transitioning from state $j$ to state $k$ or the mean number of offspring in state $k$ produced by individuals in state $j$.  Using these projection matrices and the sequence of environmental states, the population dynamic of population $i$ is given by 
\[
	X_{t+1}^i = X_t^iA_i(\xi_{t+1}, X_t).
\]
where $X_t^i$ multiplies on the left hand side of $A_i(\xi_{t+1},X_t)$ as it is a row vector.  If we define $A(\xi,X)$ to be the $n\times n$ block diagonal matrix $\mathrm{diag}(A_1(\xi,X),\dots,A_m(\xi,X))$, then the dynamics of the interacting populations are given by 
\begin{equation}\label{DYN1}
X_{t+1} = X_t A(\xi_{t+1},X_t).
\end{equation}
For these dynamics, we make the following assumptions:
\begin{enumerate}
\item[\textbf{H1:}] $\xi_1,\xi_2, \dots $ is an ergodic stationary sequence in a compact Polish space $E$ (i.e. compact, separable and completely metrizable). 
\item[\textbf{H2:}] For each $i$, $(\xi,X) \mapsto A_i(\xi,X)$ is a continuous map into the space of $n_i\times n_i$ non-negative matrices. 
\item[\textbf{H3:}] For each population $i$, the matrix $A_i$ has fixed sign structure corresponding to a primitive matrix. More precisely, for each $i$, there is a $n_i\times n_i$, non-negative, primitive matrix $P_i$ such that  the $j$-$k$-th entry of $A_i(\xi,X)$ equals zero if and only if $j$-$k$th entry $P_i$ equals zero for all $1\le j,k\le n_i$ and $(\xi,X)\in  E \times \R^n_+$. 
\item[\textbf{H4:}] There exists a compact set $S\subset \R^n_+$ such that for all $X_0\in \R^n_+$,  $X_t\in S$ for all $t$ sufficiently large. 
\end{enumerate}

Our analysis focuses on whether the interacting populations tend, in an appropriate stochastic sense, to be bounded away from extinction. Extinction of one or more population corresponds to the population state lying in the \emph{extinction set}
\[
\bS_0 = \{ x \in S : \prod_i \|x^i\|=0\}
\]
where $\|x^i\|=\sum_{j=1}^{n_i} x^{ij}$ corresponds to the $\ell^1$--norm of $x^i$. Given $X_0=x$, we define stochastic persistence in terms of the  empirical measure 
\begin{equation}\label{def:emp.meas1}
\Pi_t^x = \frac{1}{t}\sum_{s=1}^t \delta_{X_s}
\end{equation}
where $\delta_y$ denotes a Dirac measure at $y$, i.e. $\delta_{y}(A) =1$ if $y\in A$ and $0$ otherwise for any Borel set $A \subset \R_+^n$. These empirical measures are random measures  describing the distribution of the observed population dynamics up to time $t$. In particular, for any Borel set $B\subset S$, 
\[
\Pi_t^x(B)= \frac{\#\{ 1\le s \le t | X_s \in B \}}{t} 
\]
is the fraction of time that the populations spent in the set $B$. For instance, if we define  
\[
S_\eta=\{ x \in S : \|x^i\| \le \eta \mbox{ for some } i\},
\]
then $\Pi_t^x(S_\eta)$ is the fraction of time that the total abundance of some population is less than $\eta$ given $X_0=x$.

\begin{definition}\label{defn}
The model \eqref{DYN1} is \emph{stochastically persistent} if for all $\varepsilon>0$, there exists $\eta>0$ such that, with probability one,
\[
\Pi_t ^x(S_\eta) \le \varepsilon
\]
for $t$ sufficiently large and $x\in S\setminus S_0$. 
\end{definition}
The set $S_\eta$ corresponds to community states where one or more populations have a density less than $\eta$. Therefore, stochastic persistence corresponds to all populations spending an arbitrarily small fraction of time at arbitrarily low densities.

%******************************************************************
\section{Results}
%******************************************************************

\subsection{Long-term growth rates and a persistence theorem}

Understanding persistence often involves understanding what happens to each population when it is rare. To this end, we need to understand the propensity of the population to increase or decrease in the long term. Since 
\[
X_{t}^i= X_0^iA_i(\xi_1,X_0)A_i(\xi_2,X_1)\dots A_i(\xi_t, X_{t-1}),
\]
one might be interested in the long-term ``growth'' of random product of matrices
\begin{equation}\label{matrixproducts}
A_i(\xi_1,X_0)A_i(\xi_2,X_1)\dots A_i(\xi_t, X_{t-1})
\end{equation}
as $t\to\infty$. One measurement of this long-term growth rate when $X_0=x$ is the random variable
\begin{equation}\label{def:lyap}
\rr_i(x) = \limsup_{t\to\infty} \frac{1}{t} \log \| A_i(\xi_1,X_0)A_i(\xi_2,X_1)\dots A_i(\xi_t, X_{t-1})\|.
\end{equation}

Population $i$ is tending to show periods of increase when $\rr_i(x)>0$ and asymptotically decreasing when $\rr_i(x)<0$. Since, in general, the sequence  
\[
\left\{\frac{1}{t} \log \| A_i(\xi_1,X_0)A_i(\xi_2,X_1)\dots A_i(\xi_t, X_{t-1})\|
\right
\}_{t=1}^\infty
\] 
does not converge, the $\limsup_{t\to\infty}$ instead of $\lim_{t\to\infty}$ in the definition of $\rr_i(x)$ is necessary. However, as we discuss in Section \ref{sec-3.2}, the $\limsup_{t\to\infty}$ can be replaced by $\lim_{t\to\infty}$ on sets of ``full measure''. 

An expected, yet useful property of $\rr_i(x)$ is that $\rr_i(x)\le 0$ with probability one whenever $\|x^i\|>0$. In words, whenever population $i$ is present, its per-capita growth rate in the long-term is non-positive. This fact follows from $X_t^i$ being bounded above for $t \ge 0$. Furthermore, on the event of $\{\limsup_{t\to\infty} \|X^i_t\|>0\}$, we get that $\rr_i(x)=0$ with probability one. In words, if population $i$'s density infinitely often is bounded below by some minimal density, then its long-term growth rate is zero as it is not tending to extinction and its densities are bounded from above. Both of these facts are consequences of results proved in the Appendix (i.e. Proposition \ref{prop:equality-lyap}, Corollary \ref{cor_lambdanegative} and Proposition \ref{prop_lambdanul}).

Our main result extends the persistence conditions discussed in the introduction to stochastic models of interacting, structured populations. Namely, if the community increases on average when rare, then the community persists.  More formally, we prove the following theorem in the Appendix. \\

\begin{theoreme}~\label{thm:one} If there exist positive constants $p_1,\dots, p_m$ such that 
\begin{equation}\label{condition:one}
\sum_i p_i \rr_i(x)> 0 \mbox{ with probability one}
\end{equation}
for all $x\in S_0$, then the model \eqref{DYN1} is stochastically persistent. 
\end{theoreme}
For two competing species ($k=2$) that persist in isolation (i.e. $\rr_1(0)>0$ and $\rr_2(0)>0$ with probability one), inequality~\eqref{condition:one}  reduces to the classical mutual invasibility condition. To see why, consider a population state $x=(x_1,0)$ supporting species $1$. Since species $1$ can persist in isolation, Proposition \ref{prop_lambdanul} implies that $\rr_1(x)=0$ with probability one. Hence, inequality \eqref{condition:one} for this initial condition becomes $p_1 \rr_1(x)+ p_2 \rr_2(x)=p_2 \rr_2(x)>0$ with probability one for all initial conditions $x=(x_1,0)$ supporting species $1$.  Similarly, inequality \eqref{condition:one} for an initial condition $x=(0,x_2)$ supporting species $2$ becomes $\rr_1 (x)>0$ with probability one. In words, stochastic persistence occurs if both competitors have a positive per-capita growth rate when rare. A generalization of the mutual invasibility condition to higher dimensional communities is discussed at the end of the next subsection.

\subsection{A refinement using invariant measures}\label{sec-3.2}

The proof of Theorem~\ref{thm:one} follows from a more general result that we now present. For this  result, we show that one need not verify the persistence condition~(\ref{condition:one}) for all $x$ in the extinction set $S_0$. It suffices to verify the persistence condition for invariant measures of the process supported by the extinction set. 
\begin{definition}\label{def:inv.meas1}
A Borel probability measure $\mu$ on $E \times S$ is \emph{an invariant measure for the model \eqref{DYN1}}  provided that 
\begin{enumerate}
\item[(i)] $\Prob[\xi_t \in B] = \mu(B\times S)$ for all Borel sets $B\subset E$, and 
\item[(ii)] if $\Prob[(\xi_0,X_0)\in C]= \mu(C)$ for all Borel sets $C\subset E\times S$, then 
$\Prob[(\xi_t,X_t)\in C]= \mu(C)$ for all Borel sets $C\subset E\times S$ and $t\ge 0$. 
\end{enumerate} 
\end{definition}
Condition (i) ensures that invariant measure is consistent with the environmental dynamics. Condition (ii) implies that if the system initially follows the distribution of $\mu$, then it follows this distribution for all time. When this occurs, we say \emph{ $(\xi_t, X_t)$ is stationary with respect to $\mu$}. One can think of invariant measures as the stochastic analog of equilibria for deterministic dynamical systems; if the population statistics initially follow $\mu$, then they follow $\mu$ for all time.

When an invariant measure $\mu$ is statistically indecomposable, it is \emph{ergodic}. More precisely, $\mu$ is ergodic if it can not be written as a convex combination of two distinct invariant measures, i.e. if there exist $0<\alpha<1$ and two invariant measures $\mu_1,\mu_2$ such that $\mu=\alpha \mu_1 +(1-\alpha) \mu_2$, then $\mu_1=\mu_2=\mu$.

\begin{definition}\label{def:inv.rate}
If $(\xi_t, X_t)$ is stationary with respect to $\mu$, the subadditive ergodic theorem implies that  $\rr_i(X_0)$ is well-defined with probability one. Moreover, we call the expected value \[\rr_i(\mu)=\int \E[\rr_i(X_0)|X_0=x,\xi_1=\xi]\mu(d\xi,dx) \] to be \emph{long-term growth rate of species $i$ with respect to $\mu$}. When $\mu$ is ergodic, the subadditive ergodic theorem implies that $\rr_i(X_0)$ equals $\rr_i(\mu)$  for $\mu$-almost every $(X_0,\xi_1)$. \end{definition}

With these definitions, we can rephrase Theorem~\ref{thm:one} in terms of the long-term growth rates $\rr_i(\mu)$ as well as provide an alternative characterization of the persistence condition. 

\begin{theoreme}\label{thm:main}
If one of the following equivalent conditions hold
\begin{enumerate}
\item[(i)] $\rr_*(\mu) := \max_{1\le i\le m} \rr_i(\mu)>0$
for every invariant probability measure with $\mu (S_0)=1$, or 
\item[(ii)] there exist positive constants $p_1,\dots,p_m$ such that 
\[
\sum_i p_i \rr_i(\mu)>0
\]
for every ergodic probability measure with $\mu(S_0)=1$, or
\item[(iii)] there exist positive constants $p_1,\dots, p_m$ such that 
\[
\sum_i p_i \rr_i(x)> 0 \mbox{ with probability one}
\]
for all $x\in S_0$
\end{enumerate}
then the model \eqref{DYN1} is stochastically persistent. 
\end{theoreme}

With Theorem~\ref{thm:main}'s formulation of the stochastic persistence criterion, we  can introduce a generalization of the mutual invasibility condition to higher-dimensional communities. To state this condition, observe that for any ergodic, invariant measure $\mu$, there is a unique set of species $I\subset \{1,\dots,k\}$ such that $\mu(\{x: x^{ij}> 0$ for all $i\in I, 1\le j \le n_i\})=1$. In other words, $\mu$ supports the community $I$. Proposition 8.19 implies that $\rr_i(\mu)=0$ for all $i\in I$. Therefore, if $I$ is a strict subset of $\{1,\dots,k\}$ i.e. not all species are in the community $I$, then coexistence condition (ii) of Theorem~\ref{thm:main} requires that there exists a species $i\notin I$ such that $\rr_i(\mu)>0$. In other words, the coexistence condition requires that at least one missing species has a positive per-capita growth rate for any subcommunity represented by an ergodic invariant measure. While this weaker condition is sometimes sufficient to ensure coexistence (e.g. in the two species models that we examine), in general it is not as illustrated in Section 6.1. Determining, in general, when this ``at least one missing species can invade'' criterion is sufficient for stochastic persistence is an open problem.

\section{Predator-prey dynamics in auto-correlated environments} 

To illustrate the applicability of Theorems~\ref{thm:one} and~\ref{thm:main}, we apply the persistence criteria to stochastic models of predator-prey interactions,  stage-structured populations with over-compensating density-dependence, and transitive and intransitive competition in spatially heterogeneous environments. 

For unstructured populations, Theorem~\ref{thm:main} extends \citet{jmb-11}'s criteria for persistence to temporally correlated environments. These temporal correlations can have substantial consequences for coexistence as we illustrate now for a stochastic model of predator-prey interactions. In the absence of the predator, assume the prey, with density $N_t$ at time $t$, exhibits a noisy Beverton-Holt dynamic 
\begin{equation}\label{eq:bh}
N_{t+1} = \frac{R_{t+1}N_t}{1+a\,N_t}
\end{equation}
where $R_t$ is a stationary, ergodic sequence of random variables corresponding to the intrinsic fitness of the prey at time $t$, and $a>0$ corresponds to the strength of intraspecific competition. To ensure the persistence of the prey in the absence of the predator, assume $\E[\ln R_1]>0$ and $\E[\ln R_1]<\infty$. Under these assumptions, Theorem~1 of \citet{tpb-09} implies that $N_t$ converges in distribution to a positive random variable $\widehat N$ whenever $N_0>0$. Moreover, the empirical measures $\Pi^{(N,P)}_t$ with $N>0,P=0$ converge almost surely to the law $\nu$ of the random vector $(\widehat N,0)$ i.e. the probability measure satisfying $\nu(A)=\Prob[(\widehat N, 0)\in A]$ for any Borel set $A\subset \R^2_+$. 

Let $P_t$ be the density of predators at time $t$ and $\exp(-bP_t)$ be the fraction of prey that ``escape'' predation during generation $t$ where $b$ is the predator attack rate.  The mean number of predators offspring produced per consumed prey is $c$, while $s$ corresponds to the fraction of predators that survive to the next time step. The predator-prey dynamics are
\begin{equation}\label{eq:hp}
\begin{aligned}
N_{t+1} &= \frac{R_{t+1}N_t}{1+a\,N_t} \exp(-bP_t)\\
P_{t+1} &= c N_t (1-\exp(-bP_t))+s P_t.
\end{aligned}
\end{equation}
To see that (\ref{eq:hp}) is of the form of our models (\ref{DYN1}), we can expend the exponential term in the second equation. To ensures that \eqref{eq:hp} satisfies the assumptions of Theorem~\ref{thm:one}, we assume $R_t$ takes values in the half open interval $(0,R^*]$. Since  $N_{t+1} \le R_{t+1}/a \le R^*/a$ and $P_{t+1} \le c N_t + sP_t \le cR^*/a+sP_t$, $X_t=(N_t,P_t)$ eventually enters and remains in the compact set \[S=[0,R^*/a]\times [0,c R^*/(a(1-s))].\]

To apply Theorem~\ref{thm:one}, we need to evaluate $r_i((N,P))$ for all $N\ge 0, P\ge 0$ with either $N=0$ or $P=0$. 
Since $(0,P_t)$ converges to $(0,0)$ with probability one whenever $P_0\ge 0$, we have $\rr_1((0,P))=\E[\ln R_t]>0$ and $\rr_2((0,P))=\ln s<0$ whenever $P\ge 0$. Since $\Pi^{(N,0)}_t$ with $N>0$ converges almost surely to $\nu$, Proposition~\ref{prop_lambdanul} implies  $\rr_1((N,0))=0$. Moreover,
\begin{equation}
\rr_2((N,0))=\E\left[ \ln \left(cb \widehat N +s \right)\right] = \int \ln(c b x + s) \nu(dx) 
\end{equation}
By choosing $p_1=1-\varepsilon$ and $p_2= \varepsilon>0$ for $\varepsilon$ sufficiently small (e.g. $0.5\E[\ln R_t]/(\E[\ln R_t]- \ln s)$), we have $\sum_i p_i \rr_i((N,P))>0$ whenever $NP=0$ if and only if
\begin{equation}\label{eq:pred-prey}
\E\left[ \ln \left(cb \widehat N +s \right)\right]  >0.
\end{equation}
Namely, the predator and prey coexist whenever the predator can invade the prey-only system. Since $\ln(cb N +s)$ is a concave function of the prey density and the predator life history parameters $c,b,s$, Jensen's inequality implies that fluctuations in any one of these quantities decreases the predator's growth rate.

To see how temporal correlations influence whether the persistence criterion \eqref{eq:pred-prey} holds or not, consider an environment that fluctuates randomly between good and bad years for the prey. On good years, $R_t$ takes on the value $R_{good}$, while in bad years it takes on the value $R_{bad}$. Let the transitions between good and bad years be determined by a Markov chain where the probability of going from a bad year to a good year is $p$ and the probability of going from a good year to a bad year is $q$. For simplicity, we assume that $p=q$ in which case half of the years are good and half of the years are bad in the long run. Under these assumptions,  the persistence assumption $\E[\ln  R_1]>0$ for the prey is $\ln \left(R_{good} R_{bad}\right)>0.$

To estimate the left-hand side of \eqref{eq:pred-prey}, we consider the limiting cases of strongly negatively correlated environments ($p\approx 1$) and strongly positively correlated environments ($p\approx 0$). When $p\approx 1$, the environmental dynamics are nearly periodic switching nearly every other time step between good and bad years. Hence, one can approximate the stationary distribution $\widehat N$ by the positive, globally stable fixed point of 
\[
\begin{aligned}
x_{t+2}& = \frac{R_{good} x_{t+1}}{1+a x_{t+1}}\\
&= \frac{R_{good}R_{bad} x_t /(1+a x_t))}{1+ a (R_{bad} x_t/(1+ax_t))}\\
&= \frac{R_{good}R_{bad} x_t}{1+a(1+R_{bad})x_t}
\end{aligned}
\]
which is given by $\frac{R_{good}R_{bad}-1}{a(1+R_{bad})}$. Hence, if $p\approx 1$, then the distribution $\nu$ of $\widehat N$ approximately puts half of its weight on $\frac{R_{good}R_{bad}-1}{a(1+R_{bad})}$ and half of its weight on  $\frac{R_{good}R_{bad}-1}{a(1+R_{good})}$ and the persistence criterion \eqref{eq:pred-prey} is approximately
\begin{equation}\label{pp1}
\frac{1}{2} \ln \left( bc\frac{R_{good}R_{bad}-1}{a(1+R_{bad})} + s\right) + \frac{1}{2} \ln \left(bc\frac{R_{good}R_{bad}-1}{a(1+R_{good})}+s 
\right)>0.
\end{equation}

Next, consider the case that $p\approx 0$ in which there are long runs of good years and long runs of bad years. Due to these long runs, one expects that half time $\widehat N$ is near the value $(R_{good}-1)/a$ and half the time it is near the value $\max\{(R_{bad}-1)/a,0\}$. If $R_{bad}>1$, then the persistence criterion is approximately
\begin{equation}\label{pp2}
 \frac{1}{2} \ln \left(bc\frac{R_{good}-1}{a}+s \right)+
 \frac{1}{2} \ln \left(bc\frac{R_{bad}-1}{a}+s \right)>0
\end{equation}
Relatively straightforward algebraic manipulations (e.g. exponentiating the left hand sides of \eqref{pp1} and \eqref{pp2} and multiplying by $(1+R_{bad})(1+R_{good})$) show that the left hand side of \eqref{pp1} is always greater than the left hand side of \eqref{pp2}. 

\begin{bio}
Positive autocorrelations, by increasing variability in prey density, hinders predator establishment and, thereby, coexistence of the predator and prey. In contrast, negative auto-correlations by reducing variability in prey density can facilitate predator-prey coexistence (Fig.~\ref{fig:pred-prey}). \end{bio}

\begin{figure}
\includegraphics[width=6in]{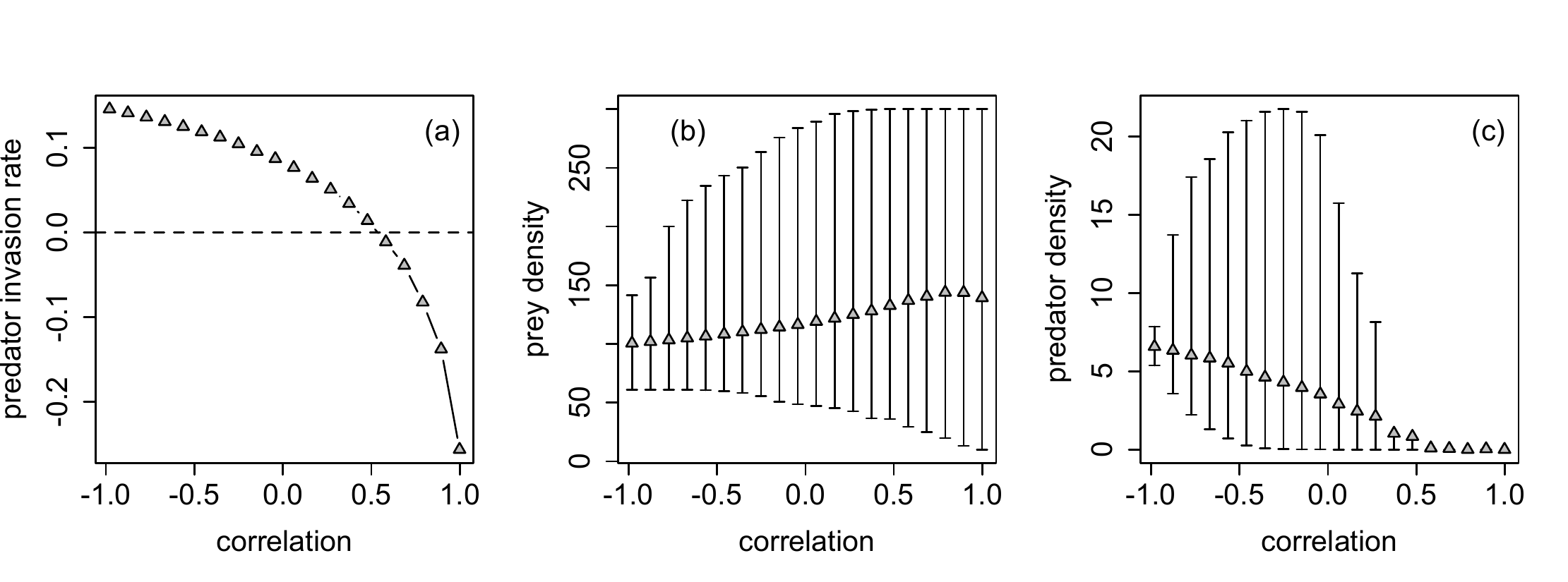}
\caption{Effect of temporal autocorrelations on predator-prey coexistence in a Markovian environment. In (a), the long-term growth rate $\rr_2((N,0))$ with $N>0$ of the predator when rare is plotted as a function function of the temporal autocorrelation between good and bad reproductive years for the prey. In (b) and (c), the mean and interquartile ranges of long-term distribution of prey and predator densities are plotted as function of the temporal autocorrelation. Parameters: $R_{good}=4$, $R_{bad}=1.1$, $a=0.01$, $c=1$, $s=0.1$, $b=0.01$.}
\label{fig:pred-prey}
\end{figure}

\section{Application to structured single species models}
For single species models with negative-density dependence, we can prove sufficient and necessary conditions for stochastic persistence. The following theorem implies that stochastic persistence occurs if the long-term growth rate $\rr_1(0)$ when rare is positive and asymptotic extinction occurs with probability one if this long-term growth rate is negative. \\

\begin{theoreme}~\label{thm:single} Assume that $m=1$ (i.e. there is one species), \textbf{H1-H4} hold and the entries of $A(\xi,x)=A_1(\xi,x)$ are non-increasing functions of $x$. If $\rr_1(0)>0$, then 
\begin{equation}\label{another}
X_{t+1} = X_t A(\xi_{t+1},X_t)
\end{equation}
is stochastically persistent. If $\rr_1(0)<0$, then $\lim_{t\to\infty} X_t = (0,0,\dots,0)$ with probability one. 
\end{theoreme} 

Our assumption that the entries $A(\xi,x)$ are non-increasing functions of $x$ ensures that $\rr_1(0)\ge \rr_1(x)$ for all $x$ which is the key fact used in the proof of Theorem~\ref{thm:single}. It remains an open problem to identify other conditions on $A(\xi,x)$ that ensure $\rr_1(0)\ge \rr_1(x)$ for all $x$.

\begin{proof}
The first statement of this theorem follows from Theorem~\ref{thm:one}. 

Assume that $\rr_1(0)<0$.  Provided that $X_0$ is nonnegative with at least one strictly positive entry, Ruelle's stochastic version of the Perron Frobenius Theorem \citep[Proposition 3.2]{ruelle-79} and the entries of $A(\xi,x)$ being non-increasing in $x$ imply 
 \[
\lim_{t\to\infty} \frac{1}{t} \log \|X_t \|  \le \lim_{t\to\infty} \frac{1}{t} \log\|X_0 A(\xi_t,0)\dots A(\xi_1,0)\| =   \rr(0)<0
\]
with probability one. Hence, $\lim_{t\to\infty} X_t=(0,\dots ,0)$ with probability one. 
\end{proof}

Theorem~\ref{thm:single} extends Theorem 1 of \citet{tpb-09} as it allows for over-compensating density dependence and makes no assumptions about differentiability of $x\mapsto A(\xi,x)$. To illustrate its utility, we apply this result to the larvae-pupue-adult model of flour beetles and a metapopulation model. 

\subsection{A stochastic Larvae-Pupae-Adult model for flour beatles}
An important, empirically validated model in ecology is the ``Larvae-Pupae-Adult'' (LPA) model which describes flour beetle population dynamics~\citep{costantino-etal-95,dennis-etal-95,costantino-etal-97}. The model keeps track of the densities $\ell_t, p_t, a_t$ of larvae, pupae, and adults at time $t$. Adults produce $b$ eggs each time step. These eggs are cannibalized by adults and larvae at rates $c_{ea}$ and $c_{el}$, respectively. The eggs escaping cannibalism become larvae. A fraction $\mu_l$ of larvae die at each time step. Larvae escaping mortality become pupae. Pupae are cannibalized by adults at a rate $c_{pa}$. Those individuals escaping cannibalism become adults. A fraction $\mu_a$ of adults survive through a time step. These assumptions result in a system of three difference equations  
\begin{equation}
\begin{aligned}
\ell_{t+1}&= b a_t \exp(-c_{el} \ell_t - c_{ea} a_t) \\
p_{t+1}& = (1-\mu_l) \ell_t \\
a_{t+1}&= \left(p_t \exp(-c_{pa} a_t)+ (1-\mu_a) a_t\right) 
\end{aligned}
\end{equation} 

Environmental fluctuations have been included in these models in at least two ways. \citet{dennis-etal-95} assumed that each stage experienced random fluctuations due to multiplicative factors $\exp(\xi_t^l),\exp(\xi_t^p),\exp(\xi_t^a)$ such that $\xi_t^i$ for $i=l,p,a$ are independent and normally distributed with mean zero i.e. on the log-scale the average effect of environmental fluctuations are accounted for by the deterministic model. Alternatively, \citet{henson-cushing-97} considered periodic fluctuations in cannibalism rates due to fluctuations in the size $V_t$ of the habitat i.e. the volume of the flour. In particular, they assumed that $c_i = \kappa_i/V_t$ for $i=ea,el,pa$, for positive constants $\kappa_i$. If we include both of these stochastic effects into the deterministic model, we arrive at the following system of random difference equations. 
\begin{equation}\label{eq:lpa}
\begin{aligned}
\ell_{t+1}&= b a_t \exp(-\kappa_{el} \ell_t/V_{t+1} - \kappa_{ea} a_t/V_{t+1}+ \xi_t^l)\\
p_{t+1}& = (1-\mu_l) \ell_t \exp(\xi_t^p)\\
a_{t+1}&= \left(p_t \exp(-\kappa_{pa} a_t/V_{t+1})+ (1-\mu_a) a_t\right)\exp( \xi_t^a)
\end{aligned}
\end{equation} 
We can use Theorem~\ref{thm:main} to prove the following persistence result.  In the case of  $\xi_t^i =0$ with probability one for $i=l,p,a$, this theorem can be viewed as a stochastic extension of Theorem 4 of \citet{henson-cushing-97} for periodic environments. 

\begin{figure}
\includegraphics[width=5in]{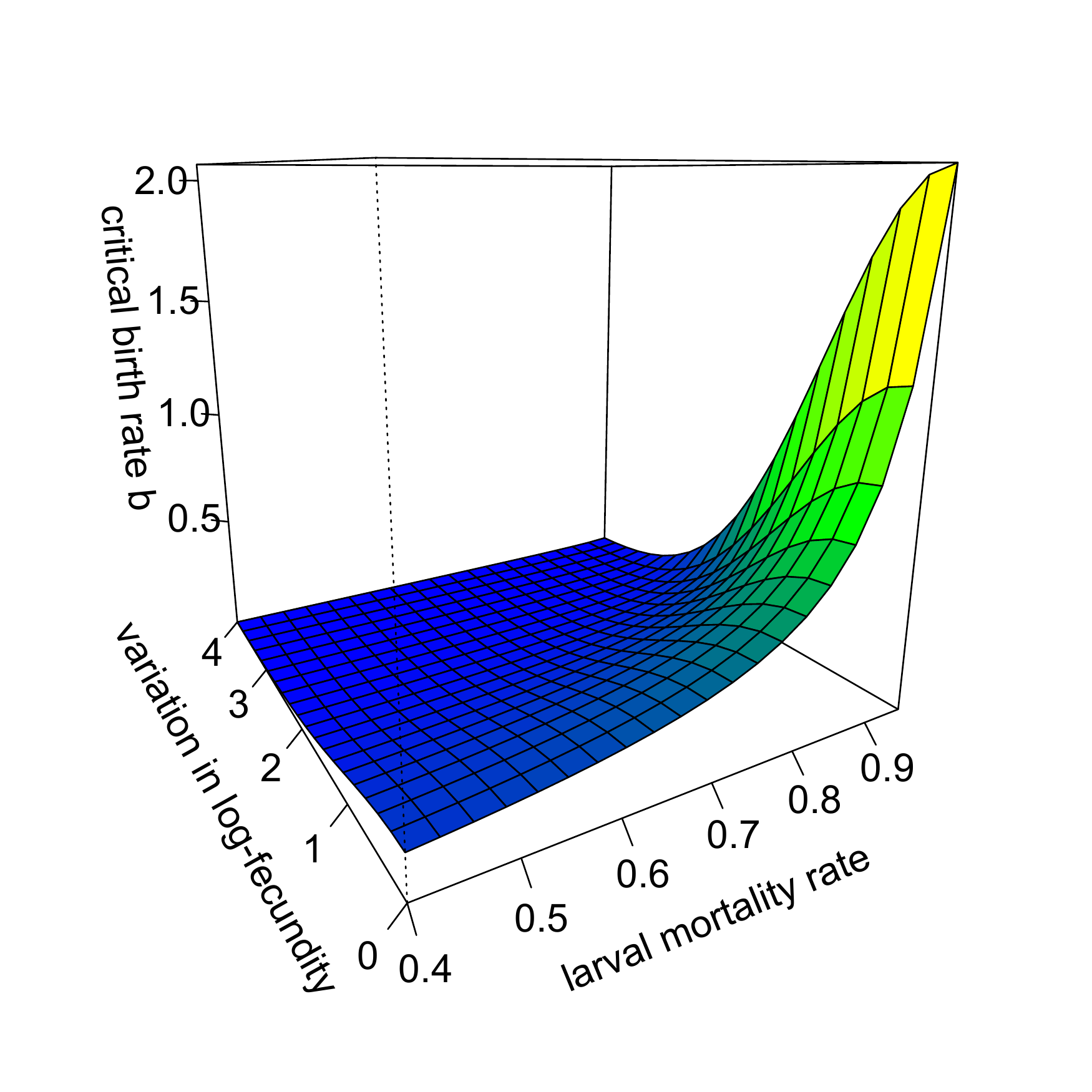}
\caption{Effects fluctuations in fecundity and larval survival on the critical birth rate $b$ required for persistence. $(\xi_t^l)$ are normally distributed with mean $0$ and variance one, $\xi_t^a=\xi_t^p=0$ for all $t$ and $\mu_a=0.1034$ (the value found in Table 1D in \citet{costantino-etal-95}).}\label{fig:LPA}
\end{figure}

\begin{theoreme}~\label{thm:lpa} Assume $c_i>0$ for $i=ea,el,pa$, $\mu_i \in (0,1)$ for $i=l,a$, $\xi_t^l$,$\xi_t^p$,$\xi_t^a$, and $V_t$ are ergodic and stationary sequences such that $\xi_t^i, \log V_t \in (-M,M)$ for $i=l,p,a$, $t\ge0$ and some $M>0$, and $(1-\mu_a)\exp(\xi^a_t) \in [0,1-\delta]$ for some $\delta>0$ with probability one. Then there exists a critical birth rate $b_{crit}>0$ such that 
\begin{description}
\item[Extinction] If $b<b_{crit}$, then $X_t=(\ell_t,p_t,a_t)$ converges almost surely to $(0,0,0)$ as $t\to\infty$. 
\item[Stochastic persistence] If $b>b_{crit}$, then the LPA model is stochastically persistent.
\end{description}
Moreover, if $\xi_t^l=\xi_t^a=\xi_t^p$ with probability one and $\E[ \xi_t^l]=0$, then 
$b_{crit}= \mu_a /(1-\mu_l)$. 
\end{theoreme}

\emph{Remark.} The assumption that $\xi_t^i$ are compactly supported formally excludes the normal distributions used by \citet{dennis-etal-95}. However, truncated normals with a very large $M$ can approximate the normal distribution arbitrarily well. The assumption $(1-\mu_a)\exp(\xi^a_t)\in [0,1-\delta]$ for some $\delta>0$ is more restrictive. However, from a biological standpoint, it is necessary as this term corresponds to the fraction of adults surviving to the next time step. None the less, we conjecture that the conclusions of Theorem~\ref{thm:lpa} hold when $\xi_t^i$ are normally distributed with mean $0$.

Theorem~\ref{thm:lpa} implies that including multiplicative noise with $\log$-mean zero has no effect on the deterministic persistence criterion when $\xi_t^l=\xi_t^p=\xi_t^a$ with probability one. However, when these random variables are not perfectly correlated,  we conjecture that this form of multiplicative noise always decreases the critical birth rate (Fig.~\ref{fig:LPA}). To provide some mathematical evidence for this conjecture, we compute a small noise approximation for the per-capita growth rate $\rr_1(\delta_0)$ when the population is rare \citep{ruelle-79,tuljapurkar-90}. Let \[
\begin{aligned}
B_t=&
\begin{pmatrix}
0&(1-\mu_l)\exp(\xi_t^p)&0\\ 
0&0& \exp(\xi_t^a) \\
b\exp(\xi_t^l)&0 &(1-\mu_a)\exp(\xi_t^a)
\end{pmatrix}\\
\end{aligned}
\]
be the linearization of the stochastic LPA model \eqref{eq:lpa} at $(L,P,A)=(0,0,0)$.  Assume that $\xi_t^i =  \varepsilon Z_t^i$ where $\E[Z_t^i]=0$ and $\E[\left(Z_t^i\right)^2]=1$. \citet[Theorem 3.1]{ruelle-79} implies that $\rr_1(0)$ is an analytic function of $\varepsilon$. Therefore, one can perform a Taylor's series expansion of $\rr_1(0)$ as function of $\varepsilon$ about the point $\varepsilon=0$. As we shall shortly show, the first non-zero term of this expansion is  of second order. Expanding $B_t$  to second order in $\varepsilon$  yields 
\[
\begin{aligned}
B_t
\approx&\underbrace{ \begin{pmatrix}
0&(1-\mu_l)&0\\ 
0&0& 1 \\
b&0 &(1-\mu_a)
\end{pmatrix}}_{=B}\left(I + \varepsilon\, \diag\{Z_t^l,Z_t^p,Z_t^a\}+ \varepsilon^2\diag\{Z_t^l,Z_t^p,Z_t^a\}^2/2\right).
\end{aligned}
\]
The entries of the second order term are positive due to the convexity of the exponential function. Hence, Jensen's inequality implies that fluctuations in $Z_t^i$ increase the mean matrix $\E[B_t]$. This observation, in and of itself, suggests that  fluctuations in $Z_t^i$ increase $\rr_1(0)$. However, to rigorously verify this assertion, let $v$ and $w$ be the left and right Perron-eigenvectors of $B$ such that $\sum_i v_i=1$ and $\sum_i v_iw_i=1$. Let $\rr_0$ be the associated Perron eigenvalue of $B$. Provided the $Z_t^i$ are independent in time, a small noise approximation for the stochastic growth rate of the random products of $B_t$ is  
\begin{equation}\label{eq:approx}
\rr_*(\delta_0)\approx \log \rr_0 + \frac{\varepsilon^2}{2}\left(\E\left[\sum_i v_i w_i \left(Z_t^i\right)^2\right]- \E\left[\left(\sum_i v_i w_i Z_t^i\right)^2\right]\right)
\end{equation}
Since the function $x\mapsto x^2$ is strictly convex and $\sum_i v_i w_i (Z_t^i)^2$ is a convex combination of $(Z_t^l)^2, (Z_t^p)^2$, and $(Z_t^a)^2$, Jensen's inequality implies  
\[
\left(\sum_i v_i w_i Z_t^i\right)^2\le \sum_i v_i w_i \left(Z_t^i\right)^2.
\]
Therefore
\[
 \E\left[\left(\sum_i v_i w_i Z_t^i\right)^2\right]\le \E\left[ \sum_i v_i w_i \left(Z_t^i\right)^2\right].
 \]
It follows that the order $\varepsilon^2$ correction term in \eqref{eq:approx} is non-negative and equals zero if and only if $Z_t^l=Z_t^p=Z_t^a$ with probability one. Therefore,  ``small'' multiplicative noise (with $\log$-mean zero) which isn't perfectly correlated across the stages increases the stochastic growth rate and, therefore, decreases the critical birth rate $b_{crit}$ required for stochastic persistence. 
 
\begin{bio}{For the LPA model, there is a critical mean fecundity, above which the population persists and below which the population goes asymptotically to extinction. Fluctuations in the log survival rates decrease the critical mean fecundity unless the log survival rates are perfectly correlated.}\end{bio}

\begin{proof}[Proof of Theorem 5.2] We begin by verifying \textbf{H1}--\textbf{H4}. \textbf{H1} and \textbf{H2} follow from our assumptions. To verify \textbf{H3}, notice that the sign structure of the nonlinear projection matrix $A_t(\xi,X)$ for \eqref{eq:lpa} is given by 
\[
C=
\begin{pmatrix}
0&1&0\\
0&0&1\\
1&0&1
\end{pmatrix}
\]
Since 
\[
C^4=
\begin{pmatrix}
1&1&1\\
1&1&2\\
2&1&3
\end{pmatrix}
\]
$A_t(\xi,X)$ has the sign structure of the primitive matrix $C$ for all $\xi,X$ and $t$. Finally, to verify \textbf{H4}, define 
\[
K=b e^{2M-1}/ \kappa_{ea}
\]
Then
\[
\ell_{t+1} \le b a_t \exp(-\kappa_{ea} a_t/V_{t+1}+ \xi_t^l)\le b a_t \exp(-\kappa_{ea} a_t \exp(-M) +M) \le K
\] 
for all $t\ge 0$. Therefore, $\ell_t\le K$ for $t\ge 1$ and  
\[
p_{t} \le \ell_{t-1} e^M \le Ke^M
\]
for all $t\ge 2$. Hence, 
\[
a_{t+1} \le p_{t}e^M + (1-\delta) a_t 
\]
for all $t \ge 2$ which implies $a_t \le Ke^{3M}/\delta$ for $t$ sufficiently large. The compact forward invariant set $S=[0,K]\times [0,Ke^M] \times [0,Ke^{3M}]/\delta$ satisfies \textbf{H4}.

At low density we get 
\[
B_t=A(\xi_t,0)=
\begin{pmatrix}
0&(1-\mu_l)\exp(\xi_t^p)&0\\ 
0&0& \exp(\xi_t^a) \\
b\exp(\xi_t^l)&0 &(1-\mu_a)\exp(\xi_t^a)
\end{pmatrix}
\]
Define $\rr(b)$ to be the dominant Lyapunov exponent of the random products of $B_1,B_2,\dots$. Note that with the notation of Theorem \ref{thm:main}, $r(b)=r_1(0)$. Theorem 3.1 of \citet{ruelle-79} implies that $\rr(b)$ is differentiable for $b>0$ and the derivative is given by (see, e.g., section 4.1 of \citet{ruelle-79})   
\[
\rr'(b)= \E\left[ \frac{v_t(b) E_{31} w_{t+1}(b)}{v_t(b) B_t(b)  w_{t+1}(b)}\right]>0
\]
where $v_t(b),w_t(b)$ are the normalized left and right invariant sub-bundles associated with $B_t(b)$ and $E_{31}$ is the matrix with $\exp(\xi_t^l)$ in the $3-1$ entry and $0$ entries otherwise. Since the numerator and denominators in the expectation are always positive,  $\rr(b)$ is a strictly increasing function of $b$. Since $\lim_{b\to 0} \rr(b)=-\infty$ and $\lim_{b\to\infty} \rr(b)=\infty$, there exists $b_{crit}>0$ such that $\rr(b)<0$ for $b<b_{crit}$ and $\rr(b)>0$ for $b>b_{crit}$.

If $b>b_{crit}$, then $\rr(b)>0$ and Theorem~\ref{thm:single} implies that \eqref{eq:lpa} is stochastically persistent. On the other hand, if $b<b_{crit}$, then $\rr(b)<0$ and Theorem~\ref{thm:single} implies that $(\ell_t,p_t,a_t)$ converges to $(0,0,0)$ with probability one as $t\to\infty$. 

The final assertion about the stochastic LPA model follows from observing that if $\xi_t^a=\xi_t^l=\xi_t^a$ with probability one for all $t$, then 
\[
B_t = 
\begin{pmatrix}
0&(1-\mu_l)&0\\ 
0&0& 1 \\
b&0 &(1-\mu_a)
\end{pmatrix}\exp(\xi_t^l)
\]
with probability one.  Hence, $\rr(b)=\log \rr_0(b) + \E[\xi_t^l]$ where $\rr_0(b)$ is the dominant eigenvalue of the deterministic matrix  
\[
\begin{pmatrix}
0&(1-\mu_l)&0\\ 
0&0& 1 \\
b&0 &(1-\mu_a)
\end{pmatrix}\]
Therefore, if $\E[\xi_t^l]=0$, then $\rr(b)=\log \rr_0(b)$. Using the Jury conditions, \citet{henson-cushing-97} showed that $\rr_0(b)>1$ if $b>\mu_a/(1-\mu_l)$ and $\rr_0(b)<1$ if $b<\mu_a/(1-\mu_l)$. Hence, when $\xi_t^l=\xi_t^p=\xi_t^a$ with probability one and $\E[\xi_t^l]=0$, $b_{crit}$ equals $\mu_a/(1-\mu_l)$ as claimed. 
\end{proof}

\subsection{Metapopulation dynamics}  Interactions between movement and spatio-temporal heterogeneities determine how quickly a population grows or declines. Understanding the precise nature of these interactive effects is a central issue in population biology receiving increasing attention from theoretical, empirical, and applied perspectives~\citep{petchy-etal-97,lundberg-etal-00,gonzalez-holt-02,schmidt-04,roy-etal-05,boyce-etal-06,hastings-botsford-06,matthews-gonzalez-07,prsb-10}. 

A basic model accounting for these interactions considers a population living in an environment with $n$  patches. Let $X^r_t$ be the number of individuals in patch $r$  at time $t$. Assuming Ricker density-dependent feedbacks at the patch scale, the fitness of an individual in patch $r$ is $\lambda^r_t \exp( - \alpha_r X_t^r)$ at time $t$, where $\lambda_t^r$ is the maximal fitness and $\alpha_r>0$ measures the strength of infraspecific competition. Let $d_{rs}$  be the fraction of the population from patch $r$ that disperse to patch $s$. Under these assumptions, the population dynamics are given by 
\begin{equation}\label{eq:model}
X^r_{t+1}= \sum_{s=1}^n d_{sr} \lambda_t^s X_t^s \exp(-\alpha_s X_t^s) \qquad r =1,\dots, n.
\end{equation}
To write this model more compactly,  let $F(X_t, \lambda_t)$ be the diagonal matrix with diagonal entries $\lambda_1 \exp(-\alpha_1 X_t^1), \dots, \lambda_n \exp(-\alpha_n X_t^n)$, and $D$ be the matrix whose $i$-$j$th entry is given by $d_{ij}$. With this notation, \eqref{eq:model} simplifies to 
\[
X_{t+1}=X_t F(X_t,\lambda_t) D
\]
If $\lambda^r_t$ are ergodic and stationary, $\lambda_t^r$ take values in a positive compact interval $[\lambda_*,\lambda^*]$ and $D$ is a primitive matrix, then the hypotheses of Theorem~\ref{thm:single} hold. In particular, stochastic persistence occurs only if $\rr_1(0)$, corresponding to the dominant Lyapunov exponent of the random matrix product $F(0,\lambda_t)D\dots F(0,\lambda_1) D$, is positive.

When populations are fully mixing (i.e. $d_{rs}=v_s$ for all $r,s$), \citet{metz-etal-83} derived a simple expression for $\rr_1(0)$ given by 
\begin{equation}\label{averaged}
\rr_1(0)=\E\left[ \log\left( \sum_{r=1}^n v_r \lambda_t^r \right)\right]
\end{equation}
i.e. the temporal log-mean of the spatial arithmetic mean. Owing to the concavity of the log function, JensenÕs inequality applied to the spatial and temporal averages in \eqref{averaged} yields
\begin{equation}\label{eq:jensen2}
\log \left( \sum_{r=1}^n v_r \E[\lambda_t^r]\right)>\rr_1(0)> \sum_{r=1}^n v_r \E[\log \lambda_t^r]
\end{equation}
The second inequality implies that dispersal can mediate persistence as $\rr_1(0)$ can be positive despite all local growth rates $\E[\log \lambda_t^r]$ being negative. Hence, populations can persist even when all patches are sinks, a phenomena that has been observed in the analysis of density-independent models and simulations of density-dependent models~\citep{jansen-yoshimura-98,bascompte-etal-02,jmb-13}. The first inequality in equation~\eqref{eq:jensen2}, however, implies that dispersal-mediated persistence for well-mixed populations requires that the expected fitness $\E[\lambda_t^ r]$ is greater than one in at least one patch. 

 For partially mixing populations for which $d_{rs}=v_s +\varepsilon \delta_{rs}$,  \citet{prsb-10} developed first-order order approximation of $\rr_1(0)$ with respect to $\varepsilon$. This approximation coupled with Theorem~\ref{thm:single} implies that temporal autocorrelations for partially mixing populations can mediate persistence even when the expected fitness $\E[\lambda_t^r]$ is less than one in all patches, a finding related to earlier work by \citet{roy-etal-05}. This dispersal mediated persistence occurs when spatial correlations are sufficiently weak, temporal fluctuations are sufficiently large and positively autocorrelated, and there are sufficiently many patches.
 
\begin{bio}{Metapopulations with density-dependent growth can stochastically persist despite all local populations being extinction prone in the absence of immigration. Temporal autocorrelations can enhance this effect. }
\end{bio}

\section{Applications to competing species in space}

The roles of spatial and temporal heterogeneity in maintaining diversity is a fundamental problem of practical and theoretical interest in population biology~\citep{chesson-00a,chesson-00b,loreau-etal-03,mouquet-loreau-03,davies-etal-05}. To examine the role of both forms of heterogeneity in maintaining diversity of competitive communities, we consider lottery-type models of $m$ competing populations in a landscape consisting of $n$ patches. For there models, competition for vacant space determines the within patch dynamics, while dispersal between the patches couples the local dynamics. After describing a general formulation of these models for an arbitrary number of species with potentially frequency-dependent interactions, we illustrate how to apply our results to case of two competing species and three competing species exhibiting an intransitive, rock-paper-scissor like dynamic. 

\subsection{Formulation of the general model}
To describe the general model, let $X_t^{ir}$ denote the fraction of patch $r$ occupied by population $i$ at time $t$.  At each time step, a fraction $\varepsilon>0$ of individuals die in each patch. The sites emptied by the dying individuals get randomly assigned to progeny in the patch. Birth rates within each patch are determined by local pair-wise interactions. Let $\xi_t^{ij}(r)$ be the ``payoff'' to strategy $i$ interacting with strategy $j$ in patch $r$ at time $t$.  Let 
\begin{equation}\label{matrix-general-model}
\Xi_t(r)=\left(\xi_t^{ij}(r)\right)_{1\le i, j\le m}
\end{equation}
be the payoff matrix for patch $r$. The total number of progeny produced by an individual playing strategy $i$ in patch $r$ is $\sum_j \xi_t^{ij} X_t^{jr}$. Progeny disperse between patches with $d_{sr}$ the fraction of progeny dispersing from patch $s$ to patch $r$. Under these assumptions, the spatial-temporal dynamics of the competing populations are given by 
\begin{equation}\label{eq:games}
X_{t+1}^{ir}=\varepsilon \frac{\sum_{s} d_{sr} \sum_{j} \xi_t^{ij}(s)X_t^{js} X_t^{is}}{\sum_s d_{sr} \sum_{j,l} \xi_t^{lj}(s) X_t^{js}X_t^{ls}} +(1-\varepsilon) X_t^{ir}.
\end{equation}
Let $A_i(\xi,X)$ be the matrix whose $s-r$ entry is given by   
\[
\varepsilon\frac{d_{sr} \sum_{j} \xi^{ij}(s) X^{js}}{\sum_{s'} d_{s'r} \sum_{j,l} \xi^{lj}(s') X^{js'}X^{ls'}}\]
for $r\neq s$, and 
\[
\varepsilon\frac{ d_{sr} \sum_{j} \xi^{ij}(s) X_t^{js}}{\sum_{s'} d_{s'r} \sum_{j,l} \xi^{lj}(s') X^{js'}X^{ls'}}+1-\varepsilon 
\]
for $r=s$.  With these definitions, \eqref{eq:games} takes on the form of our model \eqref{DYN1}.

To illustrate the insights that can be gained from a persistence analysis of these models, we consider two special cases. The first case is a spatially explicit version of \citet{chesson-warner-81}'s lottery model. The second case is a spatial version of a stochastic rock-paper-scissor game. For both of these examples, we assume that a fraction $d$ of all progeny disperse randomly to all patches and the remaining fraction $1-d$ do not disperse. Under this assumption, we get $d_{sr}=d/(m-1)$ for $s\neq r$ and $d_{ss}=1-d$. These populations are fully mixing when $d=\frac{m-1}{m}$ in which case $d_{sr}=\frac{1}{m}$ for all $s,r$.  

\subsection{A spatially-explicit lottery model} 
The lottery model of \citet{chesson-warner-81} assumes that the competing populations do not exhibit frequency dependent interactions. More specifically, the ``payoffs'' $\xi_t^{ij}(r)=\xi_t^i(r)$ for all $i,j$ are independent of the frequencies of the other species. Consequently, the model takes on a simpler form 
\begin{equation}\label{eq:lottery}
X_{t+1}^{ir}=\varepsilon \frac{\sum_{s} d_{sr} \xi_t^{i}(s) X_t^{is}}{\sum_s d_{sr} \sum_{j} \xi_t^{j}(s)X_t^{js}} +(1-\varepsilon) X_t^{ir}.
\end{equation}
where $d_{sr}=\frac{d}{m-1}$ for $r\neq s$ and $d_{ss}=1-d$. 

For two competing species (i.e. $m=2$),  the population states  $z_1 =(1,\dots,1,0,\dots, 0)$ and $z_2=(0,\dots,0,1,\dots, 1)$ correspond to only species $1$ and only species $2$ occupying the landscape, respectively. The extinction set is $S_0=\{z_1,z_2\}$. Theorem~\ref{thm:one} implies that a sufficient condition for stochastic persistence is the existence of positive weights $p_1,p_2$ such that 
\[
p_1 \rr_1(z_1)+p_2 \rr_2(z_1)>0 \mbox{ and } p_1\rr_1(z_2)+p_2\rr_2(z_2)>0
\] 
Proposition \ref{prop_lambdanul} implies that the long-term growth rate of any invariant measure, with a support bounded away from the extinction set, is equal to zero. In particular, this proposition applies to the subsystems of species $1$ and $2$, and to the Dirac measures $\delta_{z_1}$ and $\delta_{z_2}$, respectively. Therefore $\rr_1(z_1)=\rr_2(z_2)=0$ with probability one. This implies that $\rr_1(z_1)=\rr_2(z_2)=0$. Hence, the persistence criterion simplifies to 
\[
\rr_1(z_2)>0 \mbox{ and }\rr_2(z_1)>0.
\]
 In other words, persistence occurs if each species has a positive invasion rate.  
 
 To get some biological intuition from the mutual invasibility criterion, we consider the limiting cases of relatively sedentary populations (i.e. $d\approx 0$) and highly dispersive populations (i.e. $d\approx 1$). In these cases, we get explicit expressions for the realized per-capita growth rates $\rr_i(z_j)$ that simplify further for short-lived (i.e. $\varepsilon \approx 1$) and long-lived (i.e. $\varepsilon \approx 0$) species. Our analytical results are illustrated numerically in Fig.~\ref{fig:BR-ZH}.

\subsubsection{Relatively sedentary populations}
When populations are completely sedentary (i.e. $d=0$ ), the projection matrix $A_2(\xi,z_1)$ corresponding to species $2$ trying to invade a landscape monopolized by species $1$ reduces to a diagonal matrix whose $r$-th diagonal entry equals 
\[
\varepsilon \frac{\xi_t^2(r)}{\xi_t^1(r)} +1-\varepsilon
\]
The dominant Lyapunov exponent in this limiting case is given by  \[\rr_2(z_1)= \max_r \E\left[\log \left(\varepsilon \frac{\xi_t^2(r)}{\xi_t^1(r)} +1-\varepsilon\right)\right].\]  Proposition 3 from \citet{tpb-09} implies that $\rr_2(z_1)$ is a continuous function of $d$. Consequently, $\rr_2(z_1)$ is positive for small $d>0$ provided that $\E\left[\log \left(\varepsilon \frac{\xi_t^2(r)}{\xi_t^1(r)} +1-\varepsilon\right)\right]$ is strictly positive for some patch $r$. Similarly, $\rr_1(z_2)$ is positive for small $d>0$ provided that $\E\left[\log \left(\varepsilon \frac{\xi_t^1(r)}{\xi_t^2(r)} +1-\varepsilon\right)\right]$ is strictly positive for some patch $r$. Thus, coexistence for small $d>0$ occurs if 
\[
\max_r \E\left[\log \left(\varepsilon \frac{\xi_t^2(r)}{\xi_t^1(r)} +1-\varepsilon\right)\right]>0\mbox{ and }
\max_r \E\left[\log \left(\varepsilon \frac{\xi_t^1(r)}{\xi_t^2(r)} +1-\varepsilon\right)\right]>0.
\]

When $\varepsilon\approx 1$ or $\varepsilon \approx 0$, we get more explicit forms of this coexistence condition. When the populations are short-lived ($\varepsilon \approx 1$), the coexistence condition simplifies to 
$\E[\log \xi_t^1(r)]>\E[\log \xi_t^2(r)]$ and $\E[\log \xi_t^2(s)]>\E[\log \xi_t^1(s)]$ for some patches $r\neq s$. Coexistence requires that each species has at least one patch in which they have a higher geometric mean in their reproductive output. 

When the populations are long lived ($\varepsilon \approx 0$) and relatively sedentary ($d\approx 0$), the coexistence condition  is
\[
\E\left[ \frac{\xi_t^2(r)}{\xi_t^1(r)}\right]>1 \mbox{ and }  
\E\left[ \frac{\xi_t^1(s)}{\xi_t^2(s)}\right]>1 
 \]
 for some patches $r,s$. Unlike short-lived populations, it is possible that both inequalities are satisfied for the same patch. For example, when the log-fecundities $\log \xi_t^i(r)$ are independent and normally distributed with mean $\mu_i(r)$ and variance $\sigma_i^2(r)$, the coexistence conditions is
 \[
 \mu_2(r) -\mu_1(r) + \frac{\sigma_1^2(r)+\sigma_2^2(r)}{2}>1
 \]
 for some patch $r$ and 
\[
 \mu_1(s)-\mu_2(s) + \frac{\sigma_1^2(s)+\sigma_2^2(s)}{2}>1 
 \]
 for some patch $s$. Both conditions can be satisfied in the same patch $r$ provided that $\sigma_1(r)$ or $\sigma_2(r)$ is sufficiently large. 

\begin{bio}
For relatively sedentary populations, coexistence occurs if each species has a patch it can invade when rare. If the populations are also short-lived, coexistence requires that each species has a patch in which it is competitively dominant.  Alternatively, if populations are also long-lived, regional coexistence may occur if species coexist locally within a patch due to the storage effect. For uncorrelated and log-normally distributed fecundities, this within-patch storage effect occurs if the difference in the mean log-fecundities is sufficiently smaller than the net variance in the log-fecundities.  
\end{bio}

\subsubsection{Well-mixed populations}

For populations that are highly dispersive (i.e. $d=\frac{m-1}{m}$), the spatially explicit Lottery model reduces  to a spatially implicit model where 
\[
\begin{aligned}
\rr_1(z_2)& = \E \left[\log \left( \varepsilon  \frac{\sum_r \xi_t^1(r)}{\sum_r \xi_t^2(r)}+1-\varepsilon\right)\right] \mbox{ and} \\
\rr_2(z_1)& = \E \left[\log \left( \varepsilon  \frac{\sum_r \xi_t^2(r)}{\sum_r \xi_t^1(r)}+1-\varepsilon\right)\right]. \\\end{aligned}
\]
For short lived populations ($\varepsilon=1$), these long-term growth rates simplify to 
\[
\begin{aligned}
\rr_1(z_2)& = \E \left[\log  \sum_r \xi_t^1(r)\right]- \E \left[\log  \sum_r \xi_t^2(r)\right] \\
\rr_2(z_1)& = \E \left[\log  \sum_r \xi_t^2(r)\right]- \E \left[\log  \sum_r \xi_t^1(r)\right] \\
\end{aligned}
\]
Since $\rr_1(z_2)=-\rr_2(z_1)$, the persistence criterion that $\rr_1(z_2)>0$ and $\rr_2(z_1)>0$ is not satisfied generically. 

Alternatively, for long-lived populations ($\varepsilon\approx 0$), the invasion rates of well-mixed populations becomes to first-order in $\varepsilon>0$:
\[
\begin{aligned}
\rr_1(z_2) & \approx \varepsilon  \left(\E  \left[\frac{ \sum_r \xi_t^1(r)}{ \sum_r \xi_t^2(r)}\right] -1\right)\\
\rr_2(z_1)& \approx \varepsilon \left(\E  \left[\frac{ \sum_r \xi_t^2(r)}{  \sum_r \xi_t^1(r)}\right]-1\right)
\end{aligned}
\]
We conjecture that this coexistence condition is less likely to be met than the coexistence condition for relatively sedentary populations. To see why, consider a small variance approximation of these invasion rates. Assume that $\xi_t^i = \bar {\xi^i} +\eta Z^i_t(r)$
where $Z^i_t(r)$ are independent and identically distributed in $i,r$ and $\E[Z^i_t(r)]=0$ for all $i,r$. Let $\sigma^2 = \E[(Z_t^i(r))^2]$. A second order Taylor's approximation in $\eta$ yields the following approximation of the (rescaled) long-term growth rates for well-mixed populations
\begin{equation}\label{storage1}
\begin{aligned}
\E  \left[\frac{ \sum_r \xi_t^1(r)}{ \sum_r \xi_t^2(r)}\right] -1 &\approx 
\frac{\bar \xi^1}{\bar \xi^2}+\frac{\bar\xi^1\sigma^2/n}{(\bar\xi^2)^3}-1
\end{aligned}
\end{equation}
and the following approximation for relatively sedentary populations 
\begin{equation}\label{storage2}
\begin{aligned}
\max_r\E  \left[\frac{\xi_t^1(r)}{\xi_t^2(r)}\right] -1 &\approx 
\frac{\bar \xi^1}{\bar \xi^2}+\frac{\bar\xi^1\sigma^2}{(\bar\xi^2)^3}-1
\end{aligned}
\end{equation}
Since \eqref{storage2} is greater than \eqref{storage1}, persistence is more likely for relatively sedentary populations in this small noise limit. 

\begin{bio}{Short-lived and highly dispersive competitors do not satisfy the coexistence condition. Long-lived and highly-dispersive competitors may coexist. However, coexistence appears to be less likely than for sedentary populations as spatial averaging reduces the temporal variability experienced by both populations and, thereby, weakens the storage effect.}\end{bio}

\begin{figure}[t]
  \centering
  \subfloat[]{\includegraphics[width= .50  \linewidth]{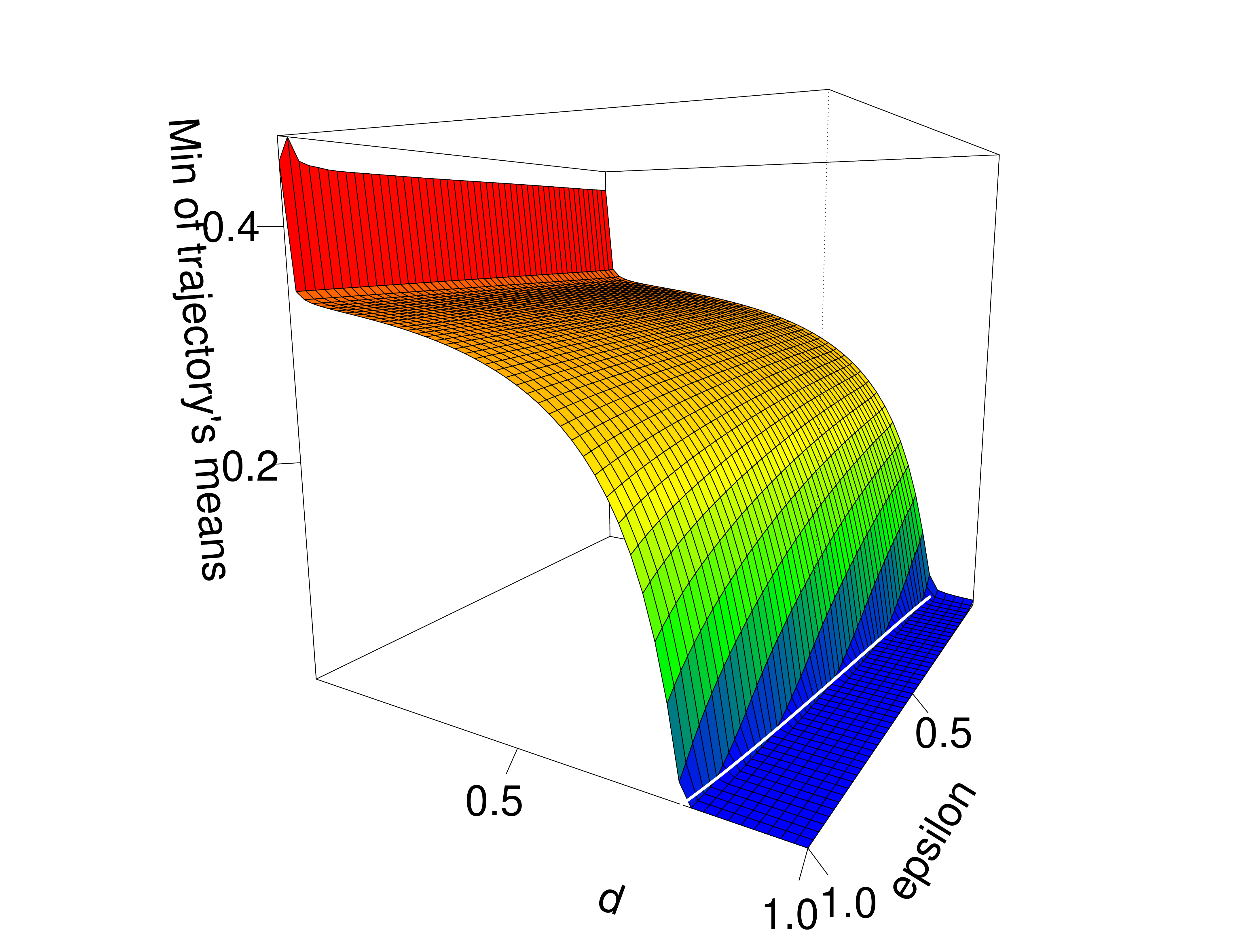}}%
  \subfloat[]{\includegraphics[width= .50  \linewidth]{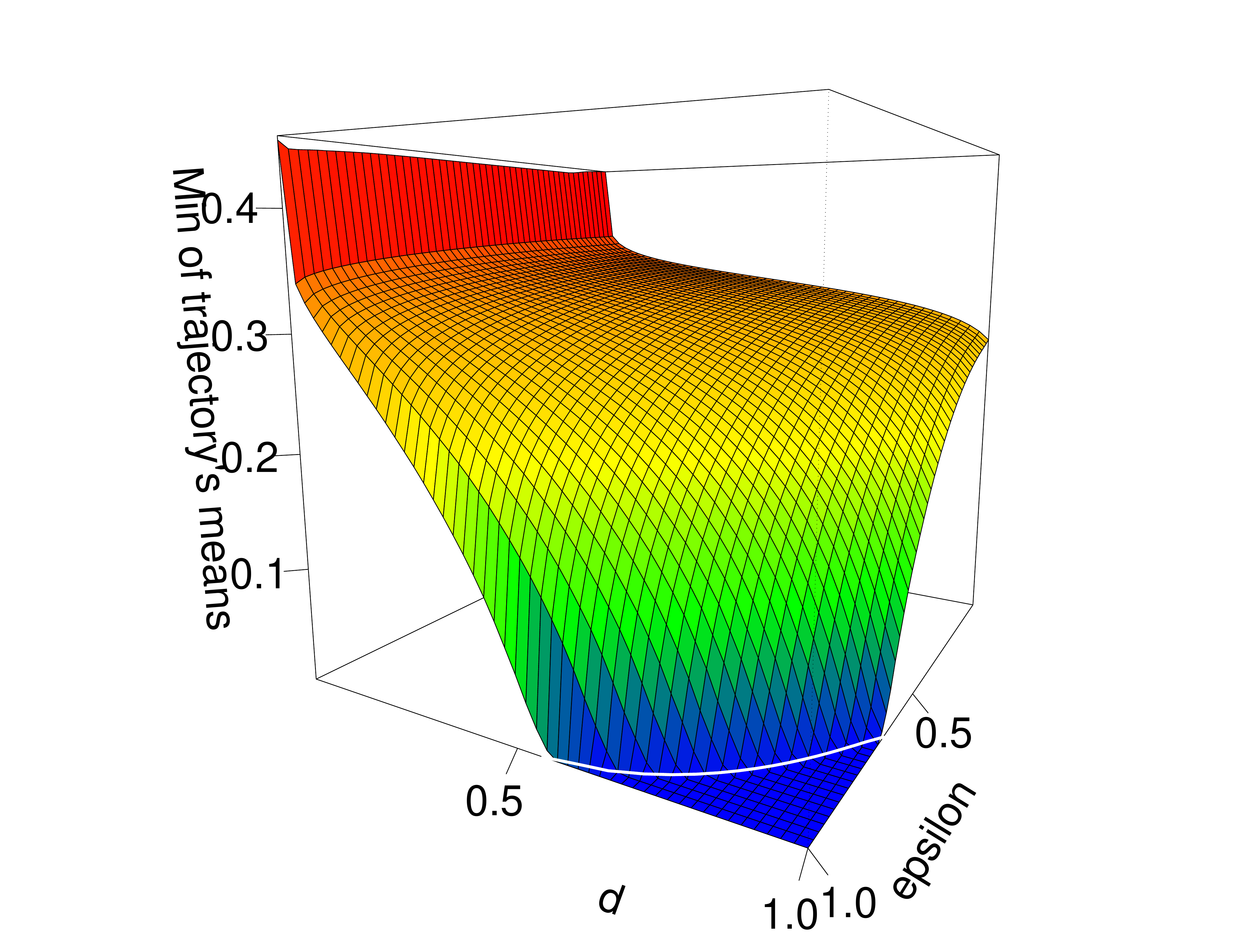}}\\
  \caption{ Effects of dispersal and survival on coexistence of two species. The log-fecundities $\xi^i$ are independent and normally distributed with means $\mu_1=(5,0,5,0,\dots,0)$, $\mu_2=(0,5,0,\dots,0)$ and variances $\sigma^2_1=\sigma^2_2=(1,\dots,1)$ for (I) and $(3,\dots,3)$ for (II). The white lines correspond to the zero-lines of the respective Lyapunov exponents.}\label{fig:BR-ZH}
\end{figure}

\subsection{The rock-paper-scissor game}

In the last few years the rock-paper-scissor game, which might initially seem to be of purely theoretical interest, has emerged as playing an important role in describing the behavior of various real-world systems. These include the evolution of alternative male mating strategies in the side-blotched lizard \emph{Uta Stansburiana} \citep{sinervo-lively-96}, the \emph{in vitro} evolution of bacterial populations \citep{kerr-etal-02,nahum-etal-11}, the \emph{in vivo} evolution of bacterial populations in mice \citep{kirkup-riley-04}, and the competition between genotypes and species in plant communities~\citep{lankau-strauss-07,cameron-etal-09}. More generally, the rock-scissors-paper game -- which is characterized by three strategies R, P and S, which satisfy the non-transitive relations: P beats R (in the absence of S), S beats P (in the absence of R), and R beats S (in the absence of P) -- serves as a simple prototype for studying the dynamics of more complicated non-transitive systems  \citep{buss-jackson-79,paquin-adams-83,may-leonard-75,jmb-97,oikos-04,vandermeer-pascual-05,allesina-levine-11}. Here, we examine a simple spatial version of this evolutionary game in a fluctuating environment. 

Let $x^1_t(r),x^2_t(r)$, and $x^3_t(r)$ be the frequencies of the rock, paper, and scissor strategies in patch $r$, respectively. All strategies in patch $r$ receive a basal payoff of $a^r_t$ at time $t$. Winners in an interaction in patch $r$ receive a payoff of $b_t^r$ while losers pay a cost $c_t^r$. Thus, the payoff matrix \eqref{matrix-general-model} for the interacting populations in patch $r$ is
\[
\Xi_t(r) = a_t^r +\begin{pmatrix} 0& -c_t^r &  b_t^r \\
b_t^r & 0& -c_t^r \\
-c_t^r & b_t^r & 0 
\end{pmatrix}.
\]
We continue to assume that the fraction of progeny dispersing from patch $s$ to patch $r$ equals $d/(m-1)$ for $s\neq r$ and $1-d$ otherwise. 

Our first result about the rock-paper-scissor model is that it exhibits a heteroclinic cycle in $S_0$ between the three equilibria $E_1=(1,\dots,1,0,\dots,0,0,\dots,0)$,
$E_2=(0,\dots,0,1,\dots,1,0,\dots,0)$ and $E_3=(0,\dots,0,0,\dots,0,1,\dots,1)$. For two vectors $x=(x_1,\dots,x_n)$, $y=(y_1,\dots,y_n)$, we write $x>y$ if $x_i\ge y_i$ for all $i$ with at least one strict inequality. 

\begin{proposition}~\label{prop:hetero} Assume $d,\varepsilon \in (0,1]$ and $a_t^r>c_t^r$, $\log a_t^r, \log c_t^r, \log b_t^r \in [-M,M]$ with probability one for some $M>0$. If $x^1_0>(0,\dots,0)$ and $x^2_0>(0,\dots,0)$ and $x_0^3=(0,\dots,0)$, then $\lim_{t\to\infty} x_t = E_2$ with probability one. If $x^1_0>(0,\dots,0)$ and $x^3_0>(0,\dots,0)$ and $x_0^2=(0,\dots,0)$, then $\lim_{t\to\infty} x_t = E_1$ with probability one. If $x^2_0>(0,\dots,0)$ and $x^3_0>(0,\dots,0)$ and $x_0^1=(0,\dots,0)$, then $\lim_{t\to\infty} x_t = E_3$ with probability one.  
\end{proposition}

\begin{proof} It suffices to prove the assertion for the case in which $x^1_0>(0,\dots,0)$ and $x^2_0>(0,\dots,0)$ and $x_0^3=(0,\dots,0)$. Let $\bone = (1,\dots,1) \in \R^n$. Our assumptions $b_t^r>0$ and $a_t^r>c_t^r>0$ imply there exists $\eta>0$ such that  $A_2(\xi_{t+1},X_t)\gg \exp(\eta) A_1(\xi_{t+1},X_t)$  with probability one. It follows that 
\[
\begin{aligned}
\limsup_{t\to\infty} \frac{1}{t} \log \|X^1_t\| & = 
\limsup_{t\to\infty} \frac{1}{t} \log \|X_0^1A_1 (\xi_1,X_0)\dots A_1(\xi_t,X_{t-1})  \|\\
&\le \limsup_{t\to\infty} \frac{1}{t} \log \|X_0^1A_2(\xi_1,X_0)\dots A_2(\xi_t,X_{t-1})\|-\eta\\ 
&= \limsup_{t\to\infty} \frac{1}{t} \log \|\bone A_2(\xi_1,X_0)\dots A_2(\xi_t,X_{t-1})\|-\eta \\
&\le -\eta
\end{aligned}
\]
where the last two lines follow from  Proposition~\ref{prop_HS} and its Corollary~\ref{cor_lambdanegative}. Hence, $\lim_{t\to\infty} \|X_t^1 \| =0$ as claimed. 
\end{proof}

Proposition~\ref{prop:hetero} implies that for any $x\in S_0$ and $1\le i \le 3$, $\rr_i(x)=\rr_i(E_j)$ for some $1\le j \le 3$. Hence, the persistence criterion of Theorem~\ref{thm:one} requires $p_1,p_2,p_3>0$ such that  
\[
\sum_i p_i \rr_i (E_j) > 0 \qquad \mbox{for all }1\le j\le 3.
\]
A standard algebraic calculation shows that this persistence criterion is satisfied if and only if 
\[
\rr_2(E_1)\rr_3(E_2)\rr_1(E_3)>
-\rr_3(E_1)\rr_1(E_2)\rr_2(E_3)
\]
i.e. the product of the positive invasion rates is greater than the absolute value of the product of the negative invasion rates. The symmetry of our model implies that all the positive invasion rates are equal and all the negative invasion rates are equal. Hence, coexistence requires
\[
\rr_2(E_1)>-\rr_3(E_1).
\]

As for the case of two competing species, we can derive more explicit coexistence criteria when the populations are relatively sedentary (i.e. $d \approx 0$) or the populations are well-mixed (i.e. $d=\frac{m}{m-1}$). For relatively sedentary populations, coexistence requires 
\[
\max_r \E\left[
\log\left(
1-\varepsilon+ \varepsilon \frac{a_t^r +b_t^r}{a_t^r}\right)
\right]> 
-\max_r \E\left[
\log\left(
1-\varepsilon+ \varepsilon \frac{a_t^r -c_t^r}{a_t^r}\right)
\right].
\] 
For long-lived populations, this coexistence criterion simplifies further to 
\[
\max_r \E\left[
 \frac{b_t^r}{a_t^r}
\right]> 
\min_r \E\left[
\frac{c_t^r}{a_t^r}
\right].
\]
Alternatively, when the populations are well-mixed, coexistence requires
\[
\E\left[
\log\left(
1-\varepsilon+ \varepsilon \frac{\sum_r a_t^r +b_t^r}{\sum_r a_t^r}\right)
\right]> 
-\max_r \E\left[
\log\left(
1-\varepsilon+ \varepsilon \frac{\sum_r a_t^r -c_t^r}{\sum_r a_t^r}\right)
\right].
\] 
For long-lived populations, this coexistence criterion simplifies further to 
\[
 \E\left[
 \frac{\sum_r b_t^r}{\sum_r a_t^r}
\right]> 
\min_r \E\left[
\frac{\sum_r c_t^r}{\sum_r a_t^r}
\right].
\]
\begin{bio}{For relatively sedentary populations, coexistence only requires that average benefits (relative to the base payoff) in one patch is greater than the average costs (relative to the base payoff) in another patch. Negative correlations between benefits $b_t^r$  and basal payoffs $a_t^r$  promote coexistence. For highly dispersive species whose base payoffs are constant in space in time (i.e. $a_t^r=a$ for all $t,r$), coexistence requires the spatially and temporally averaged benefits of interactions exceed the spatially and temporally averaged costs of interactions. 
 }\end{bio}

\section{Discussion}
Understanding the conditions that ensure the long-term persistence of interacting populations is of fundamental theoretical and practical importance in population biology. For deterministic models, coexistence naturally corresponds to an attractor bounded away from extinction. Since populations often experience large perturbation, many authors have argued that the existence of a global attractor (i.e. permanence or uniform persistence) may be necessary for long-term persistence~\citep{hofbauer-sigmund-98,smith-thieme-11}. Most populations experience stochastic fluctuations in their demographic parameters~\citep{may-73} which raises the question  \cite[pg.621]{may-73} ``How are the various usages of the term [persistence] in deterministic and stochastic circumstances related?'' Only recently has it been shown that the deterministic criteria for permanence extend naturally to criteria for stochastic persistence in stochastic difference and differential equations~\citep{benaim-etal-08,jmb-11}. These criteria assume that the populations are unstructured (i.e. no differences among individuals) and environmental fluctuations are temporally uncorrelated. However, many populations are structured as highlighted in a recent special issue in \emph{Theoretical Population Biology}~\citep{tuljapurkar-etal-12} devoted to this topic. Moreover, many environmental factors such as temperature and precipitation exhibit temporal autocorrelations~\citep{vasseur-yodzis-04}.  Here, we prove that by using long-term growth rates when rare, the standard criteria for persistence extend to models of interacting populations experiencing correlated as well as uncorrelated environmental stochasticity, exhibiting within population structure, and any form of density-dependent feedbacks. To illustrate the utility of these criteria, we applied them to persistence of predator-prey interactions in auto-correlated environments, structured populations with overcompensating density-dependence, and competitors in spatially structured environments.   

\citet{mandelbrot-82} proposed that environmental signals commonly found in nature may be composed of frequencies $f$ that scale according to an inverse power law $1/f^\beta$. With this scaling, uncorrelated (i.e. white) noise corresponds to $\beta=0$, positively auto-correlated (i.e. red or brown) noise corresponds to $\beta>0$, and negatively auto-correlated (e.g. blue) noise corresponds to $\beta<0$. Many environmental signals important to ecological processes including precipitation, mean  air temperature, degree days, and seasonal indices exhibit positive $\beta$ exponents~\citep{vasseur-yodzis-04}. Consistent with prior work on models with compensating density dependence~\citep{roughgarden-75,johst-wissel-97,petchey-00}, we found that positive autocorrelations in the maximal per-capita growth rate of species increases the long-term variability in their densities. If this species is the prey for a predatory species, we showed that this increased variability in prey densities reduced a predator's realized per-capita growth rate when rare. Hence,  positive autocorrelations may impede predator-prey coexistence. In contrast, negative autocorrelations, possibly due to a biotic feedback between the prey species and its resources, may facilitate coexistence by reducing variation in prey densities and, thereby, increase the predator's growth rate when rate. These results are qualitatively consistent with prior results that positive-autocorrelations in predator-prey systems can increase variation in prey and predator densities when they coexist~\citep{collie-spencer-94,ripa-ives-03}. Specifically, in a simulation study of predator-prey interactions in pelagic fish stocks, \citet{collie-spencer-94} found reddened noise resulted in predator-prey densities ``to shift between high and low equilibrium levels'' and, thereby, increase variability in their abundances. Similarly, using linear approximations, \citet{ripa-ives-03} found that environmental autocorrelations increased the amplitude of populations cycles. All of these results, however, stem from the per-capita growth rate of the predator being an increasing, concave function of prey density. Changes in concavity (e.g. a type-III functional response) could produce an opposing result: increased variability in prey densities may facilitate predator invasions. A more detailed analysis of this alternative is still needed.

Classical stochastic demography theory~\citep{tuljapurkar-90,boyce-etal-06} considers population growth rates in the absence of density-dependent feedbacks. Our results for populations experiencing negative-density dependence show that stochastic persistence depends on the population's long-term growth rate $\rr(0)$ when rare. Hence, applying stochastic demography theory to $\rr(0)$ provides insights into how environmental stochasticity interacts with  population structure to determine stochastic persistence. For example, a fundamental result from stochastic demography is that positive, within-year correlations between vital rates decreases $\rr(0)$ and thereby may thwart stochastic persistence, a result consist with our analysis of the stochastic LPA model for flour beetle dynamics. Stochastic demography theory also highlights that temporal autocorrelations can have subtle effects on $\rr(0)$. In particular, for a density-independent version of the metapopulation model considered here, \citet{prsb-10} demonstrated that positive temporal autocorrelations can increase the metapopulation growth rate $\rr(0)$ when rare for partially mixing populations, a prediction consistent with laboratory experiments~\citep{matthews-gonzalez-07} and earlier theoretical work~\citep{roy-etal-05}. In contrast, \citet{tuljapurkar-haridas-06} found that negative temporal autocorrelations between years with and without fires increased the realized per-capita growth rate $\rr(0)$ for models of the endangered herbaceous perennial \emph{Lomatium bradshawii}. Our results imply that these results also apply to models accounting for density-dependence. 

Spatial heterogeneity of populations has been shown theoretically and empirically to have an effect on coexistence of competitive species (see e.g. \cite{amarasekare-03} or \cite{chesson-00b} for a review). Coexistence requires species to exhibit niche differentiation that decrease the interspecific competition \citep{chesson-00a}. In a fluctuating environment, these niches can arise as differential responses to temporal variation \citep{mcgehee-armstrong-77,armstrong-mcgehee-80, chesson-00a, chesson-00b}, spatial variation \citep{may-hassell-81,chesson-00a, chesson-00b,snyder-chesson-03}, or a combination of both forms of variation \citep{chesson-85,snyder-07,snyder-08}. For the spatial lottery model where species disperse between a finite number of patches and compete for micro sites within these patches, our coexistence criterion applies, and reduces to the \emph{mutual invasibility} criterion. Although \citet{chesson-85} proved this result in the limit of an infinite number of patches with temporally uncorrelated fluctuations, our result is less restrictive as the number of patches can be small and temporal fluctuations can be autocorrelated. Using this mutual invasibility criterion, we derive explicit coexistence criteria for relatively sedentary populations and highly dispersive populations. In the former case, coexistence occurs if each species has a patch it can invade when rare. For short-lived populations, coexistence requires that each species has a patch in which it is competitively dominant. Alternatively, for long-lived populations, regional coexistence may occur if species coexist locally within a patch due to the storage effect~ \citep{chesson-warner-81,chesson-82,chesson-94} in the \emph{one patch} case. For highly dispersive populations, the coexistence criterion is only satisfied if populations exhibit overlapping generations, a conclusion consistent with \citep{chesson-85}. By providing the first mathematical confirmation of the mutual invasibility criterion for the spatial lottery model with spatial and temporal variation, our result opens the door for mathematically more rigorous investigations in understanding the relative roles of temporal variation, spatial heterogeneity, and dispersal on coexistence. 

For lottery models with three or more species, persistence criteria are more subtle and invasibility of all sub communities isn't always sufficient~\citep{may-leonard-75}. For example, in rock-paper-scissor communities where species $2$ displaces species $1$, $3$ displaces $2$ , and $1$ displaces $3$, all sub communities which consist of a single species are invasible by another, but coexistence may not occur~\citep{hofbauer-sigmund-98,tpb-13}. For the deterministic models, coexistence requires that the geometric mean of the benefits of pair-wise interactions exceeds the costs of these interactions~\citep{tpb-13}. \citet{jmb-11} and \citet{tpb-13} studied these interactions in models \emph{separately} accounting for temporal fluctuations or spatial heterogeneity. In both cases, temporal heterogeneity or spatial heterogeneity can individually promote coexistence . Here we extend these result to intransitive communities experiencing both spatial heterogeneity and temporal fluctuations, thereby unifying this prior work. Our persistence criterion reduces to: the geometric mean of the positive long-term, low-density growth rates of each species (e.g. invasion rate of rock to scissor) is greater than the geometric mean of the absolute values of the negative, long-term, low-density growth rates (e.g. invasion rate of rock to paper). For relatively sedentary populations, coexistence only requires that average benefits (relative to the base payoff) in one patch is greater than the average costs in another patch. Moreover, negative correlations between benefits and basal payoffs promote coexistence. For highly dispersive species, coexistence requires the spatially and temporally averaged benefits of interactions exceed the spatially and temporally averaged costs of interactions, assuming that base payoffs are constant in space and time.

The theory of stochastic population dynamics is confronted with many,  exciting challenges. First, our persistence criterion requires every sub community (as represented by an ergodic invariant measure supporting a subset of species) is invasible by at least one missing species. While this invasibility condition in general isn't sufficient for coexistence, understanding when it is sufficient remains a challenging open question. For example, it should be sufficient for most food chain models (see the argument for deterministic models in ~\citep{jde-00}), non-interacting prey species sharing a common predator (see the argument for deterministic models in ~\citep{jde-04}), and species competing for a single resource species. However, finding a simple criterion underlying these examples is lacking. Second, while we have provided a sufficient condition for stochastic persistence, it is equally important to develop sufficient conditions for the asymptotic exclusion of one or more species with positive probability. In light of the deterministic theory, a natural conjecture in this direction is the following: if there exist non-negative weights $p_1,\dots, p_k$ such that 
\[
\sum_i p_i \rr_i(x)<0
\]
for every population state $x$ in the extinction set $S_0$, then there exist positive initial conditions such that $X_t$ asymptotically approaches $S_0$ with positive probability. \citet{benaim-etal-08} proved a stronger version of this conjecture  for stochastic differential equation models where the diffusion term is small and the populations are unstructured. However, it is not clear whether there methods carry over to models with ``large'' noise or population structure. Another important challenge is relaxing the compactness assumption \textbf{H4} for our stochastic persistence results. While this assumption is biologically realistic (i.e. populations always have an upper limit on their size), it is theoretically inconvenient as many natural models of environmental noise have non-compact distributions (e.g. log-normal or gamma distributions). One promising approach developed by \citet{tpb-09} for structured models of single species is identifying Lyapunov-like functions that decrease on average when population densities get large. Finding sufficient conditions for ``stochastic boundedness'' is only half of the challenge, extending the stochastic persistent results to these ``stochastically bounded'' models will require additional innovations.  Finally, and most importantly, there is a desperate need to develop more tools to analytically approximate or directly compute the long-term growth rates $\rr_i(\mu)$ when rare. One promising approach is \citet{pollicott-10}'s recently derived power series representation of Lyapunov exponents.

%********************************************
\section{Proof of Theorems \ref{thm:one} and \ref{thm:main}}\label{sec-proof}

This Appendix proves Theorem~\ref{thm:main} from which Theorem~\ref{thm:one} follows.  Section \ref{partA} and \ref{partB} lead to the statement of Theorem \ref{thm:main2} which is equivalent to Theorem \ref{thm:main}. The rest of the appendix is dedicated to the proof of Theorem \ref{thm:main2}. More specifically, in section~\ref{partA}, we recast our stochastic model \eqref{DYN1} and our main hypothesis in Arnold's framework of random dynamical system~\citep{arnold-98,bhattacharya-majumdar-07}. The purpose of this recasting is to write explicitly the underlying dynamics of the matrix products (\ref{matrixproducts}) in order to use  the {Random Perron-Frobenius Theorem} (\cite{ruelle-79b}), a key element in the proof of Theorem~\ref{thm:main}. The Random Perron-Frobenius Theorem requires this underlying dynamics to be invertible which is, a priori, not the case here. Therefore, in section~\ref{partB}, we extend the underlying dynamics to an invertible dynamic on the trajectory space and state Theorem~\ref{thm:main2} which is equivalent to Theorem~\ref{thm:main}. Working in the Arnold's framework and extending the dynamic to the trajectory space requires three forms of  notation (i.e. main text, random dynamical system and trajectory space) that are summarized in Table~\ref{tab:notation}. In section~\ref{partC}, we prove basic results about the average per-capita growth rates $\rr_i$. In  section~\ref{partD}, we prove several basic results about occupational measures and their weak* limit points. These basic results are proven for the extended state space. Proposition~\ref{prop:equality-lyap} and Lemma~\ref{lemma_mes_emp} translate these results to non-extended state space.  A proof of Theorem~\ref{thm:main2} is provided in section~\ref{partE}.

\begin{table}[h!]
\caption{Notation for the probabilistic, RDS, and trajectory space formulations of the population dynamics. ``='' denotes equivalence when $\omega$ is randomly drawn from $\Omega$.}
\label{tab:notation}
\begin{tabular}{ccccc}

\hline \hline
Probabilistic formulation & & RDS formulation  & & Trajectory space\\
& &  & & formulation\\
\hline\hline

\multicolumn{1}{l}{Environmental state space}\\[0.5ex]
\hline
$E$, a Polish space &  & $\Omega$ = $E^\Z $, the space of all & & $\Omega$  \\
 & & sequences of environment state. & & \\
$\xi_t \in E$, the environment &  & $\omega= (\dots,\omega_{-1},\omega_0,\omega_1,\dots) \in \Omega$  & &$\omega=(\dots, \omega_{-1},\omega_0,\omega_1,\dots)$\\
state at time $t$ &  & is a sequence of environment states,& \\
& & i.e.  $e_t=\xi_t$ & & \\[0.5ex]
\hline
\multicolumn{1}{l}{State space} \\[0.5ex]
\hline
$\R^n_+$&$\leftarrow$ &$\Omega \times \R^n_+$ &$\leftarrow$ & $\Gamma_+ \subset \Omega \times (\R^n_+)^\Z$ \\
$x_0=p_2(\omega,x_0)$&$p_2$ &$(\omega,x_0)=\pi_0(\gamma)$ &$\pi_0$ &$\gamma = (\omega,\{x_t\}_{t\in \Z})$\\[0.5ex]
\hline
\multicolumn{1}{l}{Dynamics} \\[0.5ex]
\hline
\multicolumn{1}{l}{For one time step:} & & & &\\[0.8ex]
$X_1 = x_0A(\xi_1,x_0)$ & & $\Phi(\omega,x_0)= (\theta(\omega), x_0A(\omega_0,x_0))$ & &$\Theta$, the shift\\[0.8ex]
\multicolumn{1}{l}{For $t$ time steps:} & & & &operator on $\Gamma_+$\\[0.8ex]
$X_t=x_0A(\xi_1,x_0)\cdots A(\xi_t,X_{t-1})$ &$``="$ & $p_2(\Phi^t(\omega,x_0))$  &$=$ & $p_2(\pi_0(\Theta^t(\gamma))$\\[0.5ex]
\hline
\multicolumn{1}{l}{Empirical measures} \\[0.5ex]
\hline
$\Pi_t^x$&  &$\Lambda_t(\omega,x)$&  & $\tilde{\Lambda}_t(\gamma) $   \\
\multicolumn{1}{l}{For a Borel set $B\in \R^n_+$:} & & & &\\[0.8ex]
$\Pi_t^{x_0}(B)$&$``="$&$\Lambda_t(\omega,x_0)(\Omega,B)$&$=$&$\pi_0^*(\tilde{\Lambda}_t(\gamma))(\Omega,B) $\\
\hline
\multicolumn{1}{l}{Invariant measures} \\[0.5ex]
\hline
$\Inv := \{\mu$ satisfying Def. \ref{def:inv.meas1}$\}$ &$\leftarrow$ & $ \Inv_{\ProbQ}(\Phi)$ &$\leftarrow$ &$ \Inv_\ProbQ(\Theta) $\\[0.5ex]
$h(\mu)$ &$ h$&$\mu=\pi_0^*(\tilde{\mu})$ & $\pi_0^*$ &$\tilde{\mu}$\\[0.5ex]
$h^{-1}(\Inv)$ &$=$& $ \Inv_{\ProbQ}(\Phi)(\Omega\times \overline{V})$ &$\supset$&$\pi_0^*(\Inv_\ProbQ(\Theta)(\Gamma_+))$\\[0.5ex]
\hline
\multicolumn{1}{l}{Long-term growth rates} \\[0.5ex]
\hline
$r_i(x_0)$&$``=" $  &$r_i(\omega,x_0)$&$=$  & $r_i(\gamma) $   \\ 

\multicolumn{1}{l}{For $\mu \in \Inv$:} & &\multicolumn{1}{l}{For $\mu \in \Inv_{\ProbQ}(\Phi)$:} & &\multicolumn{1}{l}{For $\tilde{\mu} \in \Inv_\ProbQ(\Theta) $:}\\[0.8ex]
$r_i(h(\mu))$&$= $ &$r_i(\mu)$&$= $ & $r_i(\tilde{\mu}) $   \\ 
\end{tabular}

\end{table}
%\begin{table}[ht] 
%\caption{Notations} % title of Table 
%\centering      % used for centering table 
%\begin{tabular}{c c c c c c}  % centered columns (4 columns) 
%\hline\hline                        %inserts double horizontal lines 
%Main text & p. & RDS framework & p. & Trajectories space & p. \\ [0.5ex] % inserts table 
%%heading 
%\hline                    % inserts single horizontal line 
% $E$ & 4 & $\omega \in \Omega$ = $E^\Z $& 24 & - & - \\   % inserting body of the table 
%$\R^n_+$& 4 &$\Omega \times \R^n_+ $& 24 & $\gamma \in \Gamma_+ \subset (\Omega \times \R^n_+)^\Z$ & 27  \\ 
%$S_0$, $S_\eta$& 5 &-& - & $\Gamma_0$, $\Gamma_\eta$  & 27,33  \\ 
%%-& - &$p_2: \Omega \times \R^n_+ \rightarrow \R^n_+ $& 24 & $\pi_0 : \Gamma \rightarrow \Omega \times \R^n_+$ & 27  \\ 
%$X_t$& 4 &$\Phi^t(\cdot,\cdot)$& 24 & $\Theta^t(\cdot)$ & 27 \\ 
%$A_i(\xi_1,X_0)\cdots A_i(\xi_t,X_{t-1})$& 5 &$A^t_i(\omega,x)$& 25 & $A^t_i(\gamma)$ & 28 \\ 
%$\Pi_t^x$& 4 &$\Pi_t(\omega,x)$& 26 & $\tilde{\Lambda}_t(\gamma) $ & 33 \\ 
%-& - &$\mu \in \Inv_\ProbQ(\Phi)$& 25 & $\tilde{\mu} \in \Inv_\ProbQ(\Theta) $ & 28 \\ 
%$r_i(x)$& 5 &$r_i(\omega,x)$& 26 & $r_i(\gamma) $ & 28 \\ 
%$r_i(\mu)$& 5 &$r_i(\mu)$& 26 & $r_i(\tilde{\mu}) $ & 28 \\ 
%%5 & 45 & 300 & 556 \\ [1ex]       % [1ex] adds vertical space 
%\hline     %inserts single line 
%\end{tabular} 
%\label{table:nonlin}  % is used to refer this table in the text 
%\end{table} 

%*********************************************
\subsection{Random dynamical systems framework}~\label{partA}

%*********************************************
To prove our main result, it is useful to embed (\ref{DYN1}) and assumptions \textbf{H1-H4} within Arnold's general framework of random dynamical systems. Let $\Omega=E^\Z$ be the set of possible environmental trajectories, $\mathcal{F} = \mathcal{E}^\Z$ be the product $\sigma$-algebra on $\Omega$, $\theta : \Omega \mapsto \Omega$ be the shift operator defined by $\theta(\{\omega_t\}_{t \in \Z}) = \{\omega_{t+1}\}_{t \in \Z}$, and $\ProbQ$ be the probability measure on $\Omega$ satisfying
\[
\ProbQ(\{\omega \in \Omega : \omega_t \in E_0,\dots,\omega_{t+k}\in E_k\}) =\Prob(\xi_0\in E_0,\dots,\xi_k \in E_k)
\]
for any Borel sets $E_0,\dots,E_k \subset E$. Since $E$ is a Polish space, the space $\Omega$ endowed with the product topology is Polish as well. Therefore, by the Kolmogorov consistency theorem, the probability measure $\ProbQ$ is well defined, and by a theorem of \cite{rokhlin-64}, $\theta$ is ergodic with respect to $\ProbQ$. 
Randomness enters by choosing randomly a point $\omega=\{\omega_t\}_{t\in \Z} \in \Omega$ with respect to the probability distribution $\ProbQ$ and defining the environmental state at time $t$ as $\omega_t$.

In this framework, the dynamics (\ref{DYN1}) takes on the form
\begin{equation}\label{DYN2}
 \left\{
    \begin{array}{ll}
        X_{t+1}(\omega,x) = X_t(\omega,x) A(\omega_t,X_t(\omega,x)) \\
        X_0(\omega,x) =x \in \bS.
    \end{array}
\right.
\end{equation}
We  call (\ref{DYN2}), the \emph{random dynamical system determined by $(\theta,\Prob, A)$}.

Define the skew product 
\[
\begin{array}{ccl}
 \Phi : &\Omega \times \R^n_+   &\rightarrow  \Omega \times \R^n_+  \\
  & (\omega, x)  & \mapsto  (\theta(\omega), xA(\omega_0,x)) 
\end{array}
\]
associated with the dynamics (\ref{DYN2}) and define the projection maps $p_1 : \Omega \times \R^n \rightarrow \Omega$ and $p_2 : \Omega \times \R^n \rightarrow \R^n$ by $p_1(x,\omega) = \omega$ and $p_2(x,\omega) = x$. Let $\Phi^t$ denote the composition of $\Phi$ with itself $t$ times, for $t \in \N$. Remark 1.1.8 in \cite{arnold-98} implies that the random dynamical system \eqref{DYN2} is characterized by the skew product $\Phi$ and vice versa. In particular, note that $X_{t+1}(\omega,x) = p_2 \circ\Phi^{t+1}(\omega,x)$ for $x\in S$ and $\omega \in \Omega$. Working with $\Phi$ allows the use of the discrete dynamical system theory.

\begin{definition}\label{def:attractor}
A compact set $K\subset \Omega \times \R^n_+$ is a \emph{global attractor for $\Phi$} if there exists a neighborhood $V$ of $K$ such that
\begin{enumerate}
\item[(i)] for all $(\omega,x) \in  \Omega \times \R^n_+ $, there exist $T \in \N$ such that $\Phi^t(\omega,x) \in V$ for all $t \ge T$;
\item[(ii)] $\Phi(V) \subset V$ and $K=\bigcap_{t\in \N}\Phi^t(V)$.
\end{enumerate}

\end{definition}

In this random dynamical systems framework, our assumptions \textbf{H1} and \textbf{H4} take on the form
\begin{enumerate}
\item[\textbf{H1':}] $\Omega$ is a compact space, $\ProbQ$ is a Borel probability measure, and $\theta$ is an invertible map that is ergodic with respect to $\ProbQ$, i.e. for all Borel set $B \subset \Omega$, such that $\theta^{-1}(B)=B$, we have $\ProbQ(B) \in \{0,1\}$.
\item[\textbf{H4':}] There exists a global attractor $K \subset \Omega \times \R^n_+$ for $\Phi$.
\end{enumerate}
Assumptions \textbf{H2-H3} do not need to be rewritten in the new framework. Since every ergodic stationary processes on a Polish space can be described as an ergodic measure preserving transformation (Kolmogorov consistency theorem and Rokhlin theorem), assumption \textbf{H1'} is less restrictive than \textbf{H1}. Assumption \textbf{H4'} is simply restatement of assumption \textbf{H4} in the random dynamical systems framework. 

To state Theorem \ref{thm:main} in this random dynamical systems framework, we define invariant measures for the random dynamical system (\ref{DYN2}). We follow the definition given by \cite{arnold-98}. First, recall some useful definitions and notations.  Let $M$ be a metric space, and let $\PP(M)$ be the space of Borel probability measures on $M$ endowed with the weak$^*$ toplogy. If $M'$ is also a metric space and $f \colon M \to M'$ is Borel measurable, then the induced linear map $f^{*} \colon \mathcal{P}(M) \to \mathcal{P}(M')$ associates with $\nu \in \mathcal{P}(M)$ the measure $f^{\ast}(\nu) \in \mathcal{P}(M')$ defined by 
\[
f^\ast(\nu) (B)=\nu(f^{-1}(B))
\]
for all  Borel sets $B$ in $M'$.  
If $\theta \colon M \rightarrow M$  is a continuous map, a measure $\nu \in \mathcal{P}(M)$ is called \emph{$\theta$-invariant} if $\nu (\theta^{-1}(B)) = \nu(B)$ for all Borel sets $B \in M$. A set $B\subset M$ is \emph{positively invariant} if $\theta(B) \subset B$. For every positively invariant compact set $B$, let $\PPinv(B)$ be the set of all $\theta$-invariant measures supported on $B$. 

\begin{definition}\label{def:inv.meas.rds}
A probability measure $\mu$ on $\Omega \times \R^n_+$ is \emph{invariant for the random dynamical system (\ref{DYN2})} if 
\begin{enumerate}
\item[\textnormal{\textbf{(i)}}] $\mu \in \Inv(\Phi)(\Omega \times \R^n_+)$,
 \item[\textnormal{\textbf{(ii)}}] $p_1^*(\mu) = \ProbQ$, i.e. for all Borel sets $D \subset \Omega$, $\mu(D \times \R^n_+)= \ProbQ(D)$.
\end{enumerate}
For any positively invariant set $\Omega \times C$ where $C \subset \R^n_+$ is compact, $\PPinvphi(\Omega \times C)$ is the set of all measures  $\mu$ satisfying (i) and (ii) such that $\mu(\Omega \times C)=1$.
\end{definition}In words, a probability measure $\mu$ is invariant for the random dynamical system \eqref{DYN2} if it is invariant for the skew product $\Phi$ and if its first marginal is the probability $\ProbQ$ on $\Omega$.

The following result is a consequence of Theorem 1.5.10 in \cite{arnold-98}. In fact, the topology defined in his definition 1.5.3 is finer than the weak$^*$ topology on the set of all probability measures on $\Omega \times C$.
\begin{proposition}\label{prop:inv-mes-compact}
If $C \subset \R^n_+$ is a positively invariant compact set, then $\PPinvphi(\Omega \times C)$ is a nonempty, convex, compact subset of $\mathcal{P}(\Omega \times \R^n_+)$.
\end{proposition}

The main assumption in Theorem \ref{thm:main} deals with the long-term growth rates which characterize, in some sense, the long-term behavior of random matrix products (see Definition \ref{def:inv.rate}). In order to define those products in the new framework, let $\mathbf{M}_{d}(\R)$ be the set of all $d\times d$ matrices over $\R$ and consider the maps $A_i : \Omega \times S \rightarrow \mathbf{M}_{n_i}(\R)$, defined by
\[
A_i(\omega,x)=A_i(\omega_0,x).
\]

While our choice of notation here differs slightly from the main text, this choice simplifies the proof. We write  
\begin{equation}\label{def:product}
A^t_i(\omega,x) :=A_i(\omega,x)A_i(\Phi(\omega,x)) \cdots  A_i(\Phi^{t-1}(\omega,x)),
\end{equation}
with the convention that $A_i^0(\omega,x) = \mathrm{id}$, the identity matrix.

Then, for each $i\in \{1,\dots,m\}$, the asymptotic growth rate of the product (\ref{def:product}) associated with $(\omega,x) \in \Omega \times \R^n_+$ is
\[
\rr_i(\omega,x) := \limsup_{t \to \infty}\frac{1}{t}\ln \Vert A_i^t(\omega,x)\Vert,
\]
which is finite, due to assumptions \textbf{H3} and \textbf{H4'}.
According to Definition \ref{def:inv.meas.rds}, the \emph{invasion rate of species $i$ with respect to an invariant measure $\mu \in \PPinvphi$} is
\[
\rr_i(\mu):= \int_{\Omega \times \R^n_+}\rr_i(\omega, x) \mu(d\omega, dx).
\]

\begin{remark}\label{remark:samelyap}
Note that for any $x\in \R^n_+$, the random variable $\rr_i(x)$ defined by \eqref{def:lyap} is equal in distribution to the random variable $\rr_i(\cdot,x)$. Also by definition of $\ProbQ$ and $\Phi$ there is a bijection, say $h$, between the set $\PPinvphi(\Omega \times \R^n_+)$ and the set of measures defined in Definition \ref{def:inv.meas1}. Moreover the invasion rate with respect to an invariant measure is invariant by $h$, i.e. for all $\mu \in \PPinvphi$, $\rr_i(\mu)=\rr_i(h(\mu))$.
\end{remark}

Given a point $(\omega,x) \in \Omega \times \R^n_+$,  let $\Pi_t(\omega,x)$ denote the empirical occupation measure of the trajectory $\{X_s(\omega,x)\}_{s\ge0}$ at time $t$ defined by
\[
\Pi_t(\omega,x):= \frac{1}{t}\sum_{s=0}^{t-1}\delta_{X_s(\omega,x)}.
\] 
For each Borel set $B \subset \R^n_+$, the random variable $\Pi^x_t(B)$ given by (\ref{def:emp.meas1}) is equal in distribution to the random variable $\Pi_t(\cdot,x)(B)$.

For all $\eta>0$, recall that $\bS_{{\eta}}:= \{x \in \R^n_+ \ : \  \Vert x^i\Vert \le \eta \text{ for some }i \}$. We can now rephrase Theorem \ref{thm:main} in the framework of random dynamical systems.
\begin{theoreme}\label{thm:main1'}
If one of the following equivalent conditions hold
\begin{enumerate}
\item[(i)] $\rr_*(\mu):= \max_{0\le i \le m}\rr_i(\mu) >0$ for every probability measure $\mu \in \PPinvphi(\Omega \times S_0)$, or 
\item[(ii)] there exist positive constants $p_1,\dots,p_m$ such that 
\[
\sum_i p_i \rr_i(\mu)>0
\]
for every ergodic probability measure $\mu \in \PPinvphi(\Omega \times S_0)$, or
\item[(iii)] there exist positive constants $p_1,\dots,p_m$ such that 
\[
\sum_i p_i \rr_i(\omega,x)>0
\]
for every $x\in S_0$ and $\ProbQ$-almost all $\omega \in \Omega$,
\end{enumerate}
then for all $\eps >0$, there exists $\eta>0$ such that
\[
\limsup_{t \to \infty} \Pi_t(\omega,x)( \bS_{\eta}) \le \eps \ \ \text{ for }\ProbQ\text{-almost all } \omega,
\]
whenever $x \in \R^n_+ \backslash \bS_0$.
 \end{theoreme}

Remark~\ref{remark:samelyap} implies that Theorem~\ref{thm:main1'} and Theorem~\ref{thm:main} are equivalent. The remainder of the Appendix is devoted to prove Theorem~\ref{thm:main1'}.

%**************************************

 \subsection{Trajectory space\label{partB}}
 
%**************************************
The key element of the proof of Theorem \ref{thm:main1'} is Proposition \ref{prop:original_ruelle} due to \cite{ruelle-79b} in which it is crucial that the map $\Phi$ is an homeomorphism. However, the map $\Phi$ is, a priori, not invertible. To circumvent this issue, we extend the dynamics induced by $\Phi$ to an invertible map on the space of possible trajectories. Then, we state an equivalent version of Theorem \ref{thm:main1'} in this larger space that we prove in Section \ref{partE}.

By definition of the global attractor $K$, there exist a neighborhood $V$ of $p_2(K)$ in $\R^n_+$ such that $\Phi(\Omega \times V) \subset \Omega \times V$. By continuity of $\Phi$, this inclusion still holds for the closure $\overline{V}$ of $V$, i.e.
\begin{equation*}\label{eq:inclusion}
\Phi(\Omega \times \overline{V}) \subset \Omega \times \overline{V}.
\end{equation*}

This inclusion implies that, for every point $(\omega,x) \in \Omega \times \overline{V}$, there exists a sequence $\{x_t\}_{t\in \N} \subset \overline{V}^{\N}$ such that $x_0=x$, and $(\theta^{t+1}(\omega),x_{t+1})=\Phi(\theta^t(\omega),x_t)$ for all $t \ge 0$. The sequence $\{(\theta^t(\omega), x_t)\}_{t\ge 0}$ is called a \emph{$\Phi$-positive trajectory}. Note that the first coordinate of a $\Phi$-positive trajectory is characterized by $\omega$ and $\theta$. Therefore a $\Phi$-positive trajectory can be seen as a couple $(\omega, \{x_t\}_{t \ge 0})$. In order to create a \emph{past} for all those $\Phi$-positive trajectories, let us pick a point $x^*\in  \bS \backslash (\overline{V}\cup S_0)$, and consider the product space $\mathcal{T}:=\Omega \times (\overline{V} \cup \{x^*\})^{\Z}$ endowed with the product topology, and the homeomorphism $\Theta : \mathcal{T} \rightarrow \mathcal{T} $ defined by $\Theta(\omega,\{x_t\}_{t\in \Z}) = (\theta(\omega),\{x_{t+1}\}_{t\in \Z})$ and called the \emph{shift operator}. Since both $\Omega$ and $\overline{V}\cup \{x^*\}$ are compact, the space $\mathcal{T}$ is compact as well. 

Every $\Phi$-positive trajectory can be realized as an element of $\mathcal{T}$ by creating a fixed past (i.e. $x_t=x^*$ for all $t<0$). Then, define 
 \[
 \Gamma= \overline{\bigcup_{t\in \Z}\Theta^t\{ \gamma \in \mathcal{T}  : \gamma \ \text{ is a }\Phi\text{-positive trajectory}  \}}.
 \]
In words, $\Gamma$ is the adherence in $\mathcal{T}$ of the set of all shifted (by $\Theta^t$ for some $t\in \Z$) $\Phi$-positive trajectories. Since $\Gamma$ is a closed subset of the compact $\mathcal{T}$, it is compact as well. Moreover $\Gamma$ is invariant under $\Theta$, which implies that the restriction $\restr{\Theta}{\Gamma}$ of $\Theta$ on $\Gamma$ is well-defined. To simplify the presentation we still denote this restriction by $\Theta$. The projection map $\pi_0 : \Gamma \rightarrow \Omega \times \overline{V} \cup \{x^*\}$ is defined by $\pi_0(\gamma)=(\omega,x_0)$ for all $\gamma =(\omega,\{x_t\}_t) \in \Gamma$. By definition, the map $\pi_0$ is continuous and surjective. For now on, when we write $\gamma \in \Gamma$, we mean $\gamma=(\omega,\{x_t\}_{t\in \Z})$.

Define the compact set of all \emph{$\Phi$-total trajectories} as
\[
\Gamma_+:= \pi_0^{-1}(\Omega \times \overline{V}),
\]
and the compact set of \emph{$\Phi$-total-solution trajectory on the extinction set $\bS_0$} as
\[
\Gamma_0:= \pi_0^{-1}(\Omega \times \bS_0).
\]
The dynamic induced by $\Phi$ on $\Omega \times \overline{V}$ is linked to the dynamic induced by $\Theta$ on $\Gamma_+$ by the following semi conjugacy 
 \begin{equation}\label{Eq_conj}
\pi_0 \circ \Theta = \Phi \circ \pi_0.
\end{equation}
Thus, the map $\Theta$ on $\Gamma_+$ can be seen as the extension of the map $\Phi$ on $\Omega \times \overline{V}$.

In order to write an equivalent statement of Theorem~\ref{thm:main1'} with respect to the dynamics of $\Theta$, we consider a subset of the invariant measures of $\Theta$ consistent with the set $\PPinvphi(\Omega \times S)$ in the sense of Corollary~\ref{corollary:inv-theta} below. For $B \subset \Gamma$ positively $\Theta$-invariant and compact, define
\[
\PPinvtheta(B) :=\{\tilde{\mu} \in \Inv(\Theta)(B) : p_1^*\circ \pi_0^*(\tilde{\mu}) = \mathbb{Q} \}.
\]

\begin{proposition}\label{prop:mesinvtheta}
$\PPinvtheta(\Gamma_+)$ and $\PPinvtheta(\Gamma_0)$ are compact and convex subsets of $\mathcal{P}(\Gamma)$.
\end{proposition}
\begin{proof}
Since $\Gamma_+$ and $\Gamma_0$ are positively invariant compacts, $\Inv(\Theta)(\Gamma_+)$ and $\Inv(\Theta)(\Gamma_0)$ are non empty, compact and convex subsets of $\mathcal{P}(\Gamma)$. Then, since $p_1^*\circ \pi_0^*$ is continuous, $\PPinvtheta(\Gamma_+)$ (resp. $\PPinvtheta(\Gamma_0)$) is compact as closed subset of $\Inv(\Theta)(\Gamma_+)$ (resp. $\Inv(\Theta)(\Gamma_0)$). The convexity of $\PPinvtheta(\Gamma_+)$ and $\PPinvtheta(\Gamma_0)$ is a consequence of the convexity of $\Inv(\Theta)(\Gamma_+)$ and $\Inv(\Theta)(\Gamma_0)$, and the linearity of $p_1^*\circ \pi_0^*$.\end{proof}

As a consequence of equation (\ref{Eq_conj}), we have
\begin{proposition}\label{prop:inv_measure}
For every $\Theta$-invariant measure $\tilde{\mu}$ supported on $\Gamma_+$, $\pi_0^*(\tilde{\mu})$ is $\Phi$-invariant.
\end{proposition}
\begin{proof} 
Let $\tilde{\mu}$ be a $\Theta$-invariant measure supported on $\Gamma_+$. Then the measure $\pi_0^*(\tilde{\mu})$ is supported by $\Omega \times \overline{V}$. Let $B \subset \Omega \times \overline{V}$ be a Borel set. We have
\begin{eqnarray*}
\pi^*_0(\tilde{\mu})(\Phi^{-1}(B)) &=&  \tilde{\mu}(\pi_0^{-1}(\Phi^{-1}(B)))\\
&=& \tilde{\mu}(\pi_0^{-1}(\Phi^{-1}(B)) \cap \Omega \times \overline{V})\\
&=& \tilde{\mu}((\restr{\Phi}{\Omega \times \overline{V}} \circ \pi_0)^{-1}(B))\\
&=& \tilde{\mu}((\pi_0\circ \restr{\Theta}{\Gamma_+})^{-1}(B))\\
&=& \tilde{\mu}(\pi_0^{-1}(B))\\
&=& \pi^*_0(\tilde{\mu})(B).
\end{eqnarray*}
The second equality follows from the fact that the support of $\tilde{\mu}$ is included in $\Gamma_+$, and the fourth is a consequence of the conjugacy (\ref{Eq_conj}).
\end{proof}

\begin{corollary}\label{corollary:inv-theta}
$\pi_0^*(\PPinvtheta(\Gamma_+))$ is a compact and convex subset of $\PPinvphi(\Omega \times \overline{V})$.
\end{corollary}
\begin{proof} 
Since $\pi_0^*$ is continuous and linear, Proposition \ref{prop:mesinvtheta} implies that $\pi_0^*(\PPinvtheta(\Gamma_+))$ is compact and convex. Proposition \ref{prop:inv_measure} implies that $\pi_0^*(\PPinvtheta(\Gamma_+)) \subset \PPinvphi(\Omega \times \overline{V})$.\end{proof}

\begin{remark}\label{rem:invariance}
The definition of $\Theta$ and assumption \textnormal{\textbf{H3}} imply that the sets $\Gamma_0$ and $\Gamma_+\backslash \Gamma_0$ are both positively $\Theta$-invariant. Therefore every $\Theta$-invariant measure $\tilde{\mu}$ on $\Gamma_+$ can be written as a convex combination of two $\Theta$-invariant measures $\tilde{\nu}_0$ and $\tilde{\nu}_1$ such that $\tilde{\nu}_0(\Gamma_0)=1$ and $\tilde{\nu}_1(\Gamma_+\backslash\Gamma_0)=1$.
\end{remark}

In order to restate Theorem \ref{thm:main1'} in the space of trajectories, the random matrix products (\ref{def:product}) over $\Phi$ have to be rewritten as products over $\Theta$. For each $i\in \{1,\dots,m\}$, define the maps $A_i : \Gamma \rightarrow \mathbf{M}_{n_i}(\R)$ by
\[
A_i(\gamma)= \left\{
    \begin{array}{ll}
        A_i(\omega, x^*) & \text{ if } x_0=x^* \\
        A_i(\omega, x_0) &\text{ either }
    \end{array}
\right.
\]
As (\ref{def:product}), we write
\begin{equation}\label{def:cocycle}
A^t_i(\gamma) := A_i(\gamma)\cdots A_i(\Theta^{t-1}(\gamma)).
\end{equation}
The conjugacy (\ref{Eq_conj}) implies that for all $(\omega,x) \in \Omega \times \overline{V}$ and all $\gamma \in \pi_0^{-1}(\omega,x)$, we have
\begin{equation}\label{eq:cocycle}
A^t_i(\gamma)=A^t_i(\omega,x),
\end{equation}
for all $t\ge 0$.

Then the \emph{long-term growth rates} for the product (\ref{eq:cocycle}) is
\[
\rr_i(\gamma) := \limsup_{t \to \infty}\frac{1}{t} \ln \Vert A^t_i(\gamma)\Vert,
\]
and, for a $\Theta$-invariant measure $\tilde{\mu}$, the \emph{long-term growth rates} is
\[
\rr_i(\tilde{\mu}) = \int_{\Gamma}\rr_i(\gamma)d\tilde{\mu}.
\]

The following proposition shows that the long-term growth rates for the product (\ref{eq:cocycle}) defined on the trajectory space are consistent with those for the product \eqref{def:product} defined on $\Omega \times \overline{V}$.
\begin{proposition}\label{prop:equality-lyap}
For all species $i$, we have
\begin{enumerate}
\item[(i)]  $\rr_i(\omega,x) = \rr_i(\gamma)$, for all $(\omega,x) \in \Omega \times \overline{V}$ and for all $\gamma \in \pi_0^{-1}(\omega,x)$,
\item[(ii)] for all $\tilde{\mu} \in \PPinvtheta(\Gamma_+)$, $\pi_0^*(\tilde{\mu}) \in \PPinvphi(\Omega \times \overline{V})$, and \[
\rr_i(\tilde{\mu})=\rr_i(\pi^*_0(\tilde{\mu})).
\]
\end{enumerate}
 
\end{proposition}
\begin{proof} 
Assertion (i) is a consequence of equality (\ref{eq:cocycle}), and assertion (ii) is a consequence of  Corollary \ref{corollary:inv-theta}.
\end{proof}

We can now state an equivalent version of Theorem \ref{thm:main1'} on the space of trajectories $\Gamma$.
\begin{theoreme}\label{thm:main2}
If one of the following equivalent conditions hold
\begin{enumerate}
\item[(a)] $\rr_*(\tilde{\mu}):= \max_{0\le i \le m}\rr_i(\tilde{\mu}) >0$ for every probability measure $\tilde{\mu} \in \PPinvtheta(\Gamma_0)$, or 
\item[(b)] there exist positive constants $p_1,\dots,p_m$ such that 
\[
\sum_i p_i \rr_i(\tilde{\mu})>0
\]
for every ergodic probability measure $\tilde{\mu} \in \PPinvtheta(\Gamma_0)$, or
\item[(c)] there exist positive constants $p_1,\dots,p_m$ such that 
\[
\sum_i p_i \rr_i(\omega,x)>0
\]
for every $x\in S_0$ and $\ProbQ$-almost all $\omega \in \Omega$,
\end{enumerate}
then for all $\eps >0$, there exists $\eta>0$ such that
\[
\limsup_{t \to \infty} \Pi_t(\omega,x)( \bS_{\eta}) \le \eps \ \ \text{ for }\ProbQ\text{-almost all } \omega,
\]
whenever $x \in \R^n_+ \backslash \bS_0$.
 
\end{theoreme}

\begin{remark}\label{remark:equi-conditions}
 Condition (c) of Theorem \ref{thm:main2} and (iii) Theorem \ref{thm:main1'} are  equivalent, and the implications from conditions (iii) to (ii) and (ii) to (i) of Theorem~\ref{thm:main1'} are direct. The proof of Theorem \ref{thm:main2} (see section \ref{partE}) shows that (a), (b) and (c) of Theorem~\ref{thm:main2} are equivalent. Finally, condition (i) of Theorem~\ref{thm:main1'} implies condition (a) of Theorem~\ref{thm:main2} as a direct consequence of assertion (ii) of Proposition \ref{prop:equality-lyap}. Hence, Theorems~\ref{thm:main1'} and  \ref{thm:main2} are equivalent.
\end{remark}

%----------------------------------------------------
\subsection{Random Perron-Frobenius Theorem and long-term growth rates\label{partC}}
%------------------------------------------------------

In this section, we first state Proposition 3.2 of \cite{ruelle-79b} (which we call the Random Perron-Frobenius Theorem) in its original framework, and extend it to ours. We use this extension to deduce some properties on the long-term growth rates which are crucial for the proof of Theorem \ref{thm:main2}. Let $\interior \R_+^d = \{x\in \R^d_+ : \prod_i x_i >0\}$ be the interior of $\R^d_+$.

\begin{proposition}[\cite{ruelle-79b}]\label{prop:original_ruelle}
Let $\Xi$ be a compact space, $\Psi :  \Xi \rightarrow \Xi$ be an homeomorphism. Consider a continuous map $T :  \Xi \rightarrow \mathbf{M}_d(\R)$ and its transpose $T^*$ defined by $T^*(\xi)=T(\xi)^*$. Write
\[
T^t(\xi)= T(\xi)\cdots T(\Psi^{t-1}\xi),
\] 
and assume that
\begin{enumerate}
\item[\textbf{A:}] for all $\xi \in \Xi$, $T(\xi)(\R^d_+) \subset \{0\} \cup \interior \R_+^d$.
\end{enumerate}

Then there exist continuous maps $u, v : \Xi  \rightarrow \R^d_+$ with $\Vert u(\xi)\Vert =\Vert v(\xi)\Vert =1$ such that
\begin{itemize}
\item[(i)] the line bundles $E$ (resp. $F$) spanned by $u(\cdot)$ (resp. $v(\cdot)$) are such that $\R^{d} = E\bigoplus F^{\perp}$ where $b \in F(\xi)^{\perp}$ if and only if $\langle b(\xi), v(\xi)\rangle =0$.
\item[(ii)]  $E$ (resp. $F$) is $T,\Psi$-invariant (resp. $T^*,\Psi^{-1}$-invariant), i.e. $E(\Psi(\xi)) = E(\xi)T(\xi)$ and $F(\Psi \xi)T^*(\Psi\xi)=F(\xi)$, for all $\xi \in \Xi$;
\item[(iii)] there exist constants $\alpha <1$ and $C>0$ such that for all $t\ge 0$, and $\xi \in \Xi$,
\[
\Vert b T(\xi)\cdots T(\Psi^{t-1}\xi) \Vert \le C\alpha^t \Vert a T(\xi)\cdots T(\Psi^{t-1}\xi) \Vert, 
\]
for all unit vectors $a \in E(\xi), b \in F(\xi)^{\perp}$.
\end{itemize}
\end{proposition}

Our choice to called Proposition \ref{prop:original_ruelle} the \emph{Random Perron-Frobenius Theorem} is motivated by the following remark.
\begin{remark}\label{remark-perron-frobenius}
Assume that the map $T :  \Xi \rightarrow \mathbf{M}_d(\R)$ is constant, i.e. there exists $B \in  \mathbf{M}_d(\R)$ a positive matrix such that $T(\xi)= B$ for all $\xi \in \Xi$. Then Proposition \ref{prop:original_ruelle} can be restated as follows: there exist $u,v \in \R^d_+$ such that $u(\xi)=u$ and $v(\xi)=v$ for all $\xi \in \Xi$; the positive vectors $u$ and $v^*$ are respectively the right and left eigenvector of $B$ associated to its dominant eigenvalue (also called \emph{Perron eigenvalue}) $r>0$; assertion (iii) can be restated as the \emph{strong ergodic theorem of demography}. That is 
\begin{equation*}
\lim_{t \to \infty}B^tx/r^t =  v^*xu,
\end{equation*}
for all $x\in \interior \R^d_+$. Since $B^tx$ is the population at time $t$ with an initial population $x$, the interpretation of this theorem is that the eigenvector $u$ represents the \emph{stable population structure}, and the coefficients of $v$ are the \emph{reproductive values} of the population.

In Proposition \ref{prop:original_ruelle}, the stable population structure and the reproductive values can not be fixed vectors whereas long-term dynamics of the population depends on the sequence of the environment incapsulated in $\xi$. Therefore, they have to be functions of the environment, i.e. $u, v : \Xi  \rightarrow \R^d_+$. To interpret those functions, we look at the following consequence of assertion (iii) 
\begin{equation}\label{eq:stable-structure}
\lim_{t\to \infty}\frac{xT^t(\Psi^{-t}\xi)}{\Vert xT^t(\Psi^{-t}\xi) \Vert}=u(\xi),
\end{equation}
and its dual version 
\begin{equation}\label{eq:reprod-value}
\lim_{t\to \infty}\frac{T^t(\xi)y^*}{\Vert T^t(\xi)y^* \Vert}=v(\xi)^*.
\end{equation}
The former equation appears in the proof of Proposition \ref{prop:ruelle} as equation \eqref{eq:lim_ruelle}. For the sake of interpretation, assume that the environment along time has been fixed  (here $\dots,\Psi^{-1}\xi,\xi,\Psi^1\xi,\dots$). Then \eqref{eq:stable-structure} is interpreted as follows: whatever was the population a long time ago (here $x$), its structure today is given by $u(\xi)$. For equation \eqref{eq:reprod-value}, the interpretation is: whatever we assume to be the reproductive values in a long time (here $y$), the reproductive values at time $t=0$ is given by $v(\xi)$.

In applications, the environment is represented by a stationary and ergodic process $(E_t)$. Here $\xi$ represents itself a realization of this process, i.e. a possible trajectory of the environment. Therefore, there exist two stationary and ergodic processes $(U_t)$ and $(V_t)$ such that respectively $u(\xi)$ and $v(\xi)$ are realizations of them. Then equations \eqref{eq:stable-structure} and \eqref{eq:reprod-value} can be interpreted as for any initial population, in a long-term, the stage structure are given by a version of the process $(U_t)$ and the reproductive values are given by a version of $(V_t)$. 

\end{remark}
Since assumption \textbf{H2} does not directly imply assumption \textbf{A} for the map $A_i(\cdot,\cdot)$, we need to extend Ruelle's proposition to the case where
\begin{enumerate}
\item[\textbf{A1':}] for all $\xi \in \Xi$, $T(\xi)\interior \R^d_+ \subset \interior \R^d_+$, and
\item[\textbf{A2':}] there exists $s\ge1$ such that, for all $\xi \in \Xi$, $T(\xi)\cdots T(\Psi^{s-1}\xi)(\R^d_+) \subset \{0\} \cup \interior \R^d_+$.
\end{enumerate}

\begin{proposition}\label{prop:ruelle}
The conclusions of Proposition \ref{prop:original_ruelle} still hold under assumptions \textnormal{\textbf{A1'-A2'}}.
\end{proposition}

\begin{proof} 
Define the continuous map $ T' : \Xi  \rightarrow \Xi \times \mathbf{M}_d(\R)$ by
\[
T'(\xi)= T(\xi)\cdots T(\Psi^{s-1}(\xi)).
\]

By assumption \textbf{A2'}, $T'(\xi)\R^{d}_+ \subset \{0\} \cup \interior \R^{d}_+$. Therefore, Proposition \ref{prop:original_ruelle} applies to  the map $T'$ and to the homeomorphism $\Psi^s$ which give us maps $u, v : \Xi \rightarrow \R^d_+$ with $\Vert u(\xi)\Vert =\Vert v(\xi)\Vert =1$, their respective vector bundles $E(\cdot),F(\cdot)$, and some constants $C,\alpha$ verifying properties (i), (ii), and (iii).

The vector bundles $E(\cdot),F(\cdot)$ are our candidate bundles for $T$. We need only to check properties (ii) and (iii) for the map $T$ as property (i) is immediate.

We claim that 
\begin{equation}\label{eq:lim_ruelle}
\lim_{t\to \infty}\frac{xT^t(\Psi^{-t}\xi)}{\Vert xT^t(\Psi^{-t}\xi) \Vert}=u(\xi),
\end{equation}
uniformly on all compact subsets of $\R^{d}_+\setminus\{0\}$. The motivation of equation (\ref{eq:lim_ruelle}) follows from assumption \textbf{A2'} which implies that the positive cone is contracted after every interval of time of length $s$. For an interpretation of \eqref{eq:lim_ruelle}, see Remark \ref{remark-perron-frobenius}. Before we prove (\ref{eq:lim_ruelle}), we show property (ii), i.e. $E(\cdot)$ is $T,\Psi$-invariant, is a consequence (\ref{eq:lim_ruelle}). Let $y \in \interior \R^{d}_+\setminus\{0\}$, and $\xi \in \Xi$. Continuity of $T$ and equality (\ref{eq:lim_ruelle}) applied to $y$ imply
\begin{eqnarray*}
u(\xi)T(\xi)&=&\lim_{t \to \infty}\frac{yT^t(\Psi^{-t}\xi)}{\Vert yT^t(\Psi^{-t}\xi) \Vert} T(\xi)\\
&=&\lim_{t \to \infty}\frac{yT(\Psi^{-t}\xi)T^t(\Psi^{-t}(\Psi \xi))}{\Vert yT^t(\Psi^{-t}\xi) \Vert}\\
&=&u(\Psi \xi) \lim_{t\to \infty}\frac{\Vert yT(\Psi^{-t}\xi)T^t(\Psi^{-t}(\Psi \xi) \Vert}{\Vert yT^t(\Psi^{-t}\xi) \Vert},
\end{eqnarray*}
where the final line follows from (\ref{eq:lim_ruelle}) with $\xi =\Psi \xi$ and $x =yT(\Psi^{-t}\xi)/ \Vert y T(\Psi^{-t}\xi) \Vert$ which belongs to the compact $\{z \in \R^d_+ : \Vert z \Vert =1\}$ for all $t\ge0$. This proves property (ii) for $E$. The same argument for the transpose $T'^*$ implies property (ii) for $F$.

Now we prove (\ref{eq:lim_ruelle}). Let $x \in \R^{d}_+\setminus\{0\}$ with $\Vert x \Vert =1$. For every $t \ge0$, define $s_t:= t- [\frac{t}{s}]s$ where $[q]$ is the integer part of $q$. We have
\[
xT^t(\Psi^{-t}\xi)=xT^{s_t}(\Psi^{-t}\xi)T'^{[\frac{t}{s}]}(\Psi^{-t+s_t}\xi).
\]
Since $s_t\le s$ for all $t\ge0$, continuity of $T$, and assumption \textbf{A1'} imply that there is a compact $H \subset \R^d_+\setminus\{0\}$ independent of $x$ such that $xT^{s_t}(\Psi^{-t}\xi)\in H$ for all $t>0$. Then, (\ref{eq:lim_ruelle}) is a consequence of inclusion (3.2) in the proof of Proposition 3.2 in \cite{ruelle-79b} applied to the map $T'$.

It remains to check property (iii): show that there exist $\alpha', C'>0$ such that
\[
\Vert b T^t(\xi) \Vert \le C'\alpha'^t \Vert u(\xi)T^t(\xi) \Vert \text{ for all } t \ge s, \xi \in \Xi, b \in F(\xi)^{\perp}.
\]
We have
\[
b T^t(\xi)=b T^{s_t}(\xi)T'^{[\frac{t}{s}]}(\Psi^{s_t}\xi).
\]
Since $F(\cdot)$ is $T^*$-invariant, $b T^{s_t}(\xi) \in F(\Psi^{s_t}\xi)^{\perp}$ and property (iii) for $T'$ implies

\[
\frac{1}{\Vert b T^{s_t}(\xi) \Vert }\Vert b T^{s_t}(\xi)T'^{[\frac{t}{s}]}(\Psi^{s_t}\xi) \Vert \le \frac{C(\alpha^{\frac{1}{s}})^t}{\Vert u(\xi) T^{s_t}(\xi)\Vert} \Vert  u(\xi) T^{s_t}(\xi)T'^{[\frac{t}{s}]}(\Psi^{s_t}\xi) \Vert.
\]
The continuity of $T$ and $u(\cdot)$, and assumption \textbf{A1'} imply that there exist a constant $R\ge0$ such that
\[
\frac{\max\{\Vert w T^{k}(\xi)\Vert : \Vert w\Vert =1\}}{\min \{\Vert u(\xi) T^{k}(\xi)\Vert :  \xi \in \Xi \}} \le R,
\]
for all $k\le s$ and all $\xi \in \Xi$. Then property (iii) is verified with $C'=CR$ and $\alpha'=\alpha^{\frac{1}{s}}$.\end{proof}

Assumptions \textbf{H2-H3} imply that each continuous map $A_i : \Gamma  \rightarrow \mathbf{M}_{n_i}(\R)$ satisfies assumptions \textbf{A1'-A2'}. Hence Proposition \ref{prop:ruelle} applies to each continuous map $A_i$, and to the homeomorphism $\Theta$ on the compact space $\Gamma$. Then, for each of those maps, there exist row vector maps $u_i(\cdot)$, $v_i(\cdot)$, their respective vector bundles $E_i(\cdot)$, $F_i(\cdot)$, and the constant $C_i, \alpha_i >0$ satisfying properties (i), (ii), and (iii) of Proposition \ref{prop:ruelle}.

For each $i\in \{1,\dots,m\}$, define the continuous map $\bzeta_i : \Gamma \rightarrow \R$ by
\[
\bzeta_i(\gamma) := \ln \Vert u_i(\gamma)A_i(\gamma)\Vert.
\]

In the rest of this subsection, we deduce from Proposition \ref{prop:ruelle} some crucial properties of the invasions rates.

\begin{proposition}\label{prop_HS}
For all $\gamma \in \Gamma$ and every population $i$, $\rr_i(\gamma)$ satisfies the following properties:
\begin{enumerate}
\item[(i)] \[ 
\rr_i(\gamma)=\limsup_{t \to \infty}\frac{1}{t} \ln \Vert vA^t_i(\gamma)\Vert,
\]
for all $v \in \R^{n_i}_+\setminus \{0\}$ and

\item[(ii)] \[ 
\rr_i(\gamma)=\limsup_{t \to \infty}\frac{1}{t} \sum_{s=0}^{t-1}\bzeta_i(\Theta^s(\gamma)).
\]
\end{enumerate}
\end{proposition}

The proof of this proposition follows the ideas of  the proof of Proposition 1 in \cite{jde-10}.

\begin{proof} 
Let $ \gamma \in \Gamma$ be fixed. To prove the first part, we start by showing that
\begin{equation}\label{eq_ruelle}
\rr_i(\gamma) = \limsup_{t \to \infty}\frac{1}{t}\ln\Vert u_i(\gamma) A^t_i(\gamma)\Vert.
\end{equation}
Let $v\in \R^{n_i}$, $v \neq 0$. Since $\R^{n_i} = E_i(\gamma)\bigoplus F_i^{\perp}(\gamma)$, there exist a constant $a \in \R$ and a vector $w \in F_i^{\perp}(\gamma)$ such that $v = a u_i(\gamma) + w$. Then, by Proposition \ref{prop:ruelle}, we have
\begin{eqnarray*}
\Vert vA^t_i(\gamma) \Vert &\le& a \Vert u_i(\gamma) A^t_i(\gamma)\Vert + \Vert wA^t_i(\gamma) \Vert \\
  &\le&   \Vert u_i(\gamma) A^t_i(\gamma)\Vert \left( a+ C_i \alpha_i^t   \Vert w\Vert\right) .
\end{eqnarray*}
Hence,
\[
\limsup_{t \to \infty}\frac{1}{t}\ln\Vert vA^t_i(\gamma) \Vert  \le \limsup_{t \to \infty}\frac{1}{t}\ln\Vert u_i(\gamma) A^t_i(\gamma)\Vert
\]
for all $v \in \R^{n_i} \setminus \{0\}$. Since $\Vert A^t_i(\gamma)\Vert = \sup_{\Vert v\Vert =1}\Vert vA^t_i(\gamma) \Vert$, the last inequality implies that
\[
 \rr_i(\gamma) \le \limsup_{t \to \infty}\frac{1}{t}\ln\Vert u_i(\gamma) A^t_i(\gamma)\Vert \le \rr_i(\gamma),
\]
which proves the equality (\ref{eq_ruelle}).

Now, we consider positive vector $v \in \R^{n_i}_+\setminus \{0\}$. We show that the equality (\ref{eq_ruelle}) is also satisfied for $v$. We write $v=au_i(\gamma)+w$ with $a>0$ and $w\in F_i^{\perp}(\gamma)$. Proposition \ref{prop:ruelle} implies
\begin{eqnarray*}
\Vert vA^t_i(\gamma)\Vert &\ge& a \Vert u_i(\gamma) A^t_i(\gamma) \Vert - \Vert wA^t_i(\gamma) \Vert \\
  &\ge&   \Vert u_i(\gamma) A^t_i(\gamma)\Vert \left( a- C_i \alpha_i^ t   \Vert w\Vert\right).
\end{eqnarray*}
Since $a>0$,
\[
 \rr_i(\gamma) \ge \limsup_{t \to \infty}\frac{1}{t}\ln\Vert vA^t_i(\gamma) \Vert \ge \limsup_{t \to \infty}\frac{1}{t}\ln\Vert u_i(\gamma) A^t_i(\gamma)\Vert = \rr_i(\gamma),
\]
which completes the proof of assertion (i).

The second assertion results directly from the first assertion and the following equalities:
\begin{eqnarray*}
\ln \Vert u_i(\gamma) A_i^{t+1}(\gamma)\Vert &=& \ln \Vert u_i(\gamma)A^t_i(\gamma)A_i(\Theta^t(\gamma)) \Vert \\
&=& \ln \left\Vert u_i(\Theta(\gamma)^t)A_i(\Theta^t(\gamma)) \right\Vert  \left\Vert u_i(\gamma)A^t_i(\gamma) \right\Vert \\
&=& \bzeta_i(\Theta^t(\gamma)) + \ln \left\Vert u_i(\gamma)A^t_i(\gamma)\right\Vert.
\end{eqnarray*}
The second step is a consequence of the invariance of the line bundle $E_i$.\end{proof}

Recall that $\Gamma_+= \pi_0^{-1}(\Omega \times \overline{V})$ and $\Gamma_0= \pi_0^{-1}(\Omega \times \bS_0)$.

\begin{corollary}\label{cor_lambdanegative}
For all $\gamma \in \Gamma_+ \setminus \Gamma_0$, and every $i\in \{1,\dots,m\}$, 
\[
\rr_i(\gamma) \le 0.
\]
\end{corollary}

\begin{proof}
Fix $i \in \{1,\dots,m\}$, and $\gamma \in \Gamma_+\backslash \Gamma_0$ with $(\omega,x) :=\pi_0(\gamma)$. By definition of $\Gamma_+\backslash \Gamma_0$,  $x^i\in \R^{n_i}_+$ and $x^i \ne 0$. We have
\begin{eqnarray*}
x^iA^t_i(\gamma) &=& x^iA_i(\gamma)\cdots A_i(\Theta^{t-1}\gamma)\\
&=& x^iA_i(\omega,x)\cdots A_i(\Phi^{t-1}(\omega,x))\\
&=& p_2(\Phi^t(\omega,x)),
\end{eqnarray*}
where the second equality is a consequence of (\ref{eq:cocycle}), and the third one follows from the definition of the cocycle $\Phi$. Assumption \textbf{H4'} implies that there exists $T>0$ such that $p_2(\Phi^t(\omega,x))$ belongs to the compact set $\overline{V}$ for all $t\ge T$, which implies that there exists $R>0$ such that $\Vert x^iA^t_i(\gamma) \Vert \le R$ for all $t\ge T$. Assertion (i) of Proposition \ref{prop_HS} applied to $v= x^i$ concludes the proof.\end{proof}

Now we give some properties of the invasion rate with respect to a $\Theta$-invariant probability measure.

\begin{proposition}\label{prop_intlambda}
The invasion rate of each population $i$ with respect to an $\Theta$-invariant measure $\tilde{\mu}$ satisfies the following property:
\[
\rr_i(\tilde{\mu}) = \int_{\Gamma} \bzeta_i(\gamma) d\tilde{\mu}.
\]
\end{proposition}

\begin{proof} 
This result is a direct consequence of property (ii) of Proposition \ref{prop_HS} and the Birkhoff's Ergodic Theorem applied to the continuous maps $\Theta$ and $\bzeta$. \end{proof}

\begin{proposition}\label{prop_lambdanul}
Let $\tilde{\mu}$ be a $\Theta$-invariant measure. If $\tilde{\mu}$ is supported by $\Gamma_+ \backslash \Gamma_0$, then $\rr_i(\tilde{\mu})=0$ for all $i \in \{1, \dots,m\}$.
\end{proposition}

\begin{proof} 
Let $\tilde{\mu}$ be such a probability measure. Fix $i \in \{1, \dots,m\}$, and define the set $\Gamma^{i,{\eta}}:= \{ \gamma \in \Gamma_+ : \Vert p_2(\pi_0(\Theta^t(\gamma)))^i\Vert> \eta\}$. By assumption on the measure $\tilde{\mu}$, there exists a real number $\eta^{*}>0$ such that $\tilde{\mu}(\Gamma^{i,{\eta}}) >0$ for all $\eta<\eta^*$.

The Poincar\'e  recurrence theorem applies to the map $\Theta$, and implies that for each $\eta<\eta^*$,
\begin{equation}\label{eq:poincare}
\tilde{\mu}(\{\gamma \in \Gamma^{i,{\eta}} \vert \ \Theta^t(\gamma) \in  \Gamma^{i,{\eta}} \text{ infinitely often }\})=1.
\end{equation}
Recall that the conjugacy (\ref{Eq_conj}) implies that for every $\gamma \in \Gamma_+$ with $\pi_0(\gamma)=(\omega,x) \in \Omega \times \overline{V}\backslash S_0$, we have 
\begin{eqnarray*}
p_2(\pi_0(\Theta^t(\gamma)))^i&=&p_2(\Phi^t(\pi_0(\gamma)))^i\\
&=&x^iA_i^t(\gamma).
\end{eqnarray*}
Then, equality (\ref{eq:poincare}) means that for $\tilde{\mu}$-almost all $\gamma \in \Gamma^{i,{\eta}}$ with $0<\eta <\eta^*$,  $ \Vert x^iA_i^t(\gamma) \Vert > \eta$ infinitely often. Therefore, Proposition \ref{prop_HS} (i), applied to $v=x^i$, implies that $\rr_i(\gamma)=\limsup_{t \to \infty}\frac{1}{t} \ln \Vert x^iA^t_i(\gamma)\Vert \ge 0$ for $\tilde{\mu}$-almost all $\gamma \in \Gamma^{i,{\eta}}$, with  $\eta<\eta^*$. Hence $\rr_i(\gamma)\ge 0$ for $\tilde{\mu}$-almost all $\gamma \in  \bigcup_{n \ge \frac{1}{\eta^*}} \Gamma^{i,{1/n}} = \Gamma_+ \backslash \Gamma_0$. Corollary \ref{cor_lambdanegative} completes the proof. \end{proof}

%*************************************
\subsection{Properties of the empirical occupation measures\label{partD}}

%*************************************
Given a trajectory $\gamma \in \Gamma_+$, the \emph{empirical occupation measure} at time $t\in \N$ of $\{\Theta^s(\gamma)\}_{s\ge 0}$ is 
\[
\tilde{\Lambda}_t(\gamma) := \frac{1}{t} \sum_{s=0}^{t-1}\delta_{\Theta^s(\gamma)},
\]
and given a point $(\omega, x) \in \Omega \times \overline{V}$, the \emph{empirical occupation measure} at time $t\in \N$ of $\{\Phi^s(\omega,x)\}_{s\ge 0}$ is 
\[
\Lambda_t(\omega,x) := \frac{1}{t} \sum_{s=0}^{t-1}\delta_{\Phi^s(\omega,x)}.
\]
In this way, $\Lambda_t(\omega,x)(\Omega \times B)=\Pi_t(\omega,x)(B)$ for every Borel subset $B \subset \overline{V}$, and $x \in \overline{V}$.  
 
The dynamics $\Theta$ and $\Phi$ being semi-conjugated by $\pi_0$, their respective empirical occupation measures are semi-conjugated by $\pi_0^*$ as follows. 
\begin{lemma}\label{lemma_mes_emp}
Let $\gamma \in \Gamma_+$. Then for all $t\ge 0$ we have
\[
\pi_0^*(\tilde{\Lambda}_t(\gamma))=\Lambda_t(\pi_0(\gamma)).
\]
\end{lemma}
\begin{proof}
Let $B \subset \Omega \times \overline{V}$ be a Borel set, and $\gamma \in \Gamma_+$. Then we have
\begin{eqnarray*}
\pi_0^*(\tilde{\Lambda}_t(\gamma))(B) &=& \tilde{\Lambda}_t(\gamma)(\pi_0^{-1}(B))\\
&=& \frac{1}{t} \sum_{s=0}^{t-1}\delta_{\Theta^s(\gamma)}(\pi_0^{-1}(B))\\
&=& \frac{1}{t} \sum_{s=0}^{t-1}\delta_{\Phi^s(\pi_0(\gamma))}(B)\\
&=& \Lambda_t(\pi_0(\gamma))(B).
\end{eqnarray*}
The third equality is a consequence of the semi conjugacy (\ref{Eq_conj}) .
\end{proof}
 
 \begin{proposition}\label{prop_walters}
 There exists $\tilde{\Omega}$ with $\ProbQ(\tilde{\Omega})=1$ such that for all $\gamma \in \pi_0^{-1}(\tilde{\Omega}\times \overline{V})$, the set of all weak$^*$ limit point of the family of probability measures $\{\tilde{\Lambda}_t(\gamma)\}_{t\in \N}$ is a non-empty subset of $\PPinvtheta(\Gamma_+)$.
\end{proposition}
  
 \begin{proof}
Since $\ProbQ$ is ergodic (assumption \textbf{H4'}), Birkhoff's Ergodic Theorem implies that there exists a subset $\tilde{\Omega} \subset \Omega$ such that $\ProbQ(\tilde{\Omega})=1$, and for all $\omega \in \tilde{\Omega}$,
\begin{equation}\label{prop_wal_eq1}
\lim_{t \to \infty}\frac{1}{t} \sum_{s=0}^{t-1}\delta_{\theta^s(\omega)} = \ProbQ
\end{equation}
(in the weak$^*$ topology). Let $(\omega,x) \in \tilde{\Omega}\times \overline{V}$ and $\gamma \in \pi_0^{-1}(\omega,x) \subset \Gamma_+$. For all $t \in \N$, we have
 \begin{equation}\label{prop_wal_eq2}
 p_1^*\circ \pi_0^*(\tilde{\Lambda}_t(\gamma)) = \frac{1}{t} \sum_{s=0}^{t-1}\delta_{\theta^s(\omega)}.
\end{equation}
Since $\Gamma_+$ is positively $\Theta$-invariant and compact, the set of all weak$^*$ limit point of the family of probability measures $\{\tilde{\Lambda}_t(\gamma)\}_{t\in \N}$ is a non-empty subset of  $\mathcal{P}(\Gamma_+)$. Since the maps $p_1$ and $\pi_0$ are continuous, equalities (\ref{prop_wal_eq1}) and (\ref{prop_wal_eq2}) imply that $p_1^*\circ \pi_0^*(\tilde{\mu})=\ProbQ$. Moreover, Theorem 6.9 in \cite{walters-82} implies that $\tilde{\mu}$ is $\Theta$-invariant. Therefore, $\tilde{\mu} \in \PPinvtheta(\Gamma_+)$, which concludes the proof.  \end{proof}

Recall that $\bS_{{\eta}}=\{ x \in S : \| x^i \| \le \eta\mbox{ for some } i\}$, and define the subset $\Gamma_{\eta}:= \pi_0^{-1}(\Omega \times S_{\eta})$.
\begin{proposition}\label{lemma1}
If condtion \textnormal{(a)} of Theorem \ref{thm:main2} is satisfied, then for all $\eps >0$ there exists $\eta^*>0$ such that
 \[
 \tilde{\mu}( \Gamma_{\eta})<\eps,
 \]
for all $\eta<\eta^*$ and all $\tilde{\mu} \in \PPinvtheta(\Gamma_+ \backslash \Gamma_0)$.
\end{proposition}
 
\begin{proof}
If false, there exist $\eps>0$ and a sequence of measures $\{\tilde{\mu}_n\}_{n\in \N} \subset \PPinvtheta(\Gamma_+ \backslash \Gamma_0)$ such that $\tilde{\mu}_n( \Gamma_{1/n}) > \eps$ for all $n \ge 1$. By Proposition \ref{prop:mesinvtheta}, let $\tilde{\mu} \in \PPinvtheta(\Gamma_+)$ be a weak$^*$ limit point of the sequence $\{\tilde{\mu}_n\}_{n\in \N}$.  Proposition \ref{prop_lambdanul} implies that $\rr_*(\tilde{\mu}_n)=0$ for all $n \ge 0$. Proposition \ref{prop_intlambda} and weak$^*$ convergence imply that $0 = \lim_{n\to \infty}\rr_i(\tilde{\mu}_n) = \rr_i(\tilde{\mu})$ for all $i$. Hence,  $\rr_*(\tilde{\mu})=0$. The Portmanteau theorem (see e.g.  Theorem 2.1. in \cite{billingsley-99}) applied to the closed set $\Gamma_{1/n}$ implies that for all $n\ge1$,
\begin{eqnarray*}
\tilde{\mu}(\Gamma_{1/n}) &\ge& \liminf_{m\to \infty}\tilde{\mu}_m(\Gamma_{1/n})\\
&\ge& \liminf_{m\to \infty}\tilde{\mu}_m(\Gamma_{1/m})\\
&\ge& \eps.
\end{eqnarray*}
Therefore $\tilde{\mu}(\Gamma_0) = \tilde{\mu}(\cap_n \Gamma_{1/n})\ge \eps$. Remark \ref{rem:invariance} implies there exist $\alpha > 0$ such that $\tilde{\mu} = \alpha \tilde{\nu}_0 + (1-\alpha)\tilde{\nu}_1$ where $\tilde{\nu}_j$ are $\Theta$-invariant probability measures satisfying $\tilde{\nu}_0(\Gamma_0) =1$ and $\tilde{\nu}_1(\Gamma_+ \backslash  \Gamma_0)=1$. By Proposition \ref{prop_lambdanul}, $\rr_i(\tilde{\nu}_1) = 0$ for all $i \in \{1, \dots,k \}$. Condition (a) implies that $\rr_*(\tilde{\nu}_0)  >0$, in which case $0=\rr_*(\tilde{\mu}) = \alpha \rr_*(\tilde{\nu}_0)>0$ which is a contradiction. \end{proof}
\vspace{0.5cm}

%---------------------------------------------------------------
\subsection{Proof of Theorem \ref{thm:main2}\label{partE}}
%------------------------------------------------------------------
First, we show that condition (a) of Theorem \ref{thm:main2} implies that for all $\eps >0$, there exists $\eta>0$ such that
\[
\limsup_{t \to \infty} \Pi_t(\omega,x)( \bS_{\eta}) \le \eps \ \ \text{ for }\ProbQ\text{-almost all } \omega,
\]
whenever $x \in \R^n_+ \backslash \bS_0$. Second, we prove the equivalence of conditions (a), (b) and (c).

Let $\tilde{\Omega} \subset \Omega$ be defined as in Proposition \ref{prop_walters}. Choose $(\omega',x') \in \tilde{\Omega} \times \R^n_+ \backslash \bS_0$. By definition of the set $\overline{V}$, there exists a time $T\ge 0$ such that $\Phi^t(\omega',x') \in \Omega \times \overline{V}$, for all $t\ge T$. Choose $\gamma \in \pi_0^{-1}(\Phi^T(\omega',x')) \subset \Gamma_+ \backslash \Gamma_0$. Since $\mu$ is a weak$^*$ limit point of the family $\{\Lambda_t(\Phi^T(\omega',x'))\}_{t\ge 0}$ if and only if it is a weak$^*$ limit point of the family $\{\Lambda_t(\omega',x')\}_{t\ge 0}$, we do not loss generality by considering $\{\Lambda_t(\Phi^T(\omega',x'))\}_{t\ge 0}$. Since $\Omega \times \overline{V}$ is compact, the set of all weak$^*$ limit points of the family of probability measures $\{\Lambda_t(\Phi^T(\omega',x'))\}_{t\in \N}$ is a non-empty subset of $\mathcal{P}(\Omega \times \overline{V})$. Let $\mu = \lim_{k\to \infty} \Lambda_{t_k}(\omega,x)$ be such a weak$^*$ limit point. Since $\Gamma_+$ is positively $\Theta$-invariant and compact, passing to a subsequence if necessary, there exists $\tilde{\mu} = \lim_{k \to \infty}\tilde{\Lambda}_{t_k}(\gamma) \in \mathcal{P}(\Gamma_+)$. By Proposition \ref{prop_walters}, $\tilde{\mu} \in \PPinvtheta(\Gamma_+)$. Furthermore by Lemma \ref{lemma_mes_emp} and continuity of $\pi_0$, $\pi_{0}^*(\tilde{\mu})=\mu$. Hence, Proposition \ref{prop_intlambda}, the continuity of the map $\bzeta$, and property (ii) of Proposition \ref{prop_HS}, imply the following equalities for all $i$: 
\begin{eqnarray*}
\rr_i(\tilde{\mu})&=&\int_{\Gamma} \bzeta(\eta) d\tilde{\mu}(\eta) \\
&=& \lim_{k \to \infty}\frac{1}{t_k}\sum_{s=0}^{t_k-1}\bzeta(\Theta^s(\gamma))\\
&\le&\rr_i(\gamma).
\end{eqnarray*}
Hence, by Corollary \ref{cor_lambdanegative}, 
\[
\rr_i(\tilde{\mu}) \le 0, \ \text{ for all } i.
\]
Remark \ref{rem:invariance} implies there exists $\alpha \ge 0$ such that $\tilde{\mu} = \alpha \tilde{\nu}_0 + (1-\alpha)\tilde{\nu}_1$ where $\tilde{\nu}_j$ are invariant probability measure satisfying $\tilde{\nu}_0(\Gamma_0) =1$ and $\tilde{\nu}_1(\Gamma_+ \backslash \Gamma_0)=1$. By Proposition \ref{prop_lambdanul}, $\rr_i(\tilde{\nu}_1) = 0$ for all $i \in \{1, \dots,k \}$. Condition (a) implies $\rr_*(\tilde{\nu}_0) >0$. Therefore $\alpha$ must be zero, i.e. $\tilde{\mu}(\Gamma_+ \backslash \Gamma_0)=1$. Fix $\eps >0$. By Proposition \ref{lemma1} there exists $\eta^*>0$ such that
 \[
 \tilde{\mu}( \Gamma_{\eta})<\eps, \ \ \forall \eta<\eta^*,
 \]
 which implies
 \[
 \mu(\Omega \times \bS_{\eta})<\eps, \ \ \forall \eta<\eta^*.
 \]
Since $\eta^*$ does not depend on $\mu$, we have
\[
\limsup_{t \to \infty}\Lambda_t(\omega',x')( \Omega \times \bS_{\eta}) < \eps, \ \ \forall \eta<\eta^*,
\]
for all $x' \in \R^n_+\backslash \bS_0$ and $\omega' \in \tilde{\Omega}$, which concludes the first part of the proof.

Next, we show the equivalence of conditions (a) and (b). We need the following version of the minimax theorem (see, e.g., \cite{simmons-98}):

\begin{theoreme}[Minimax theorem] Let $A,B$ be Hausdorff topological vector spaces and let
$\mathcal{L} : A \times B \to \R$ be a continuous bilinear function. Finally, let $E$ and $F$ be nonempty,
convex, compact subsets of $A$ and $B$, respectively. Then
\[
\min_{a\in E}\max_{b\in F}\mathcal{L}(a, b) = \max_{b\in F}\min_{a\in E} \mathcal{L}(a,b).
\]
\end{theoreme}

We have that
\[
\min_{\tilde{\mu}} \max_i \rr_i (\tilde{\mu}) = \min_{\tilde{\mu} }\max_{p} \sum_i p_i \rr_i(\tilde{\mu})
\]
where the minimum is taken over $\tilde{\mu} \in  \PPinvtheta(\Gamma_0)$ and the maximum over $p\in \Delta:= \{p \in \R^m_+ \ : \ \sum_i p_1=1\}$.
Define $A$ to be the dual space to the space of bounded continuous functions from $\Gamma_0$ to $\R$ and
define $B=\R^m$. Let $E=\Delta$, and $D = \PPinvtheta(\Gamma_0) \subset A$ which is nonempty, convex and compact by Proposition \ref{prop:mesinvtheta}. Let $\mathcal{L}: A \times B \rightarrow \R$ the bilinear function defined by $\mathcal{L}(\tilde{\mu}, p):=\sum_i p_i \rr_i(\tilde{\mu})$. Proposition \ref{prop_intlambda} implies that $\mathcal{L}$ is continuous. With these choices, the Minimax theorem  implies that
\begin{equation}\label{eq:minmax}
\min_{\tilde{\mu}} \max_i \rr_i (\tilde{\mu}) = \max_{p\in \Delta} \min_{\tilde{\mu}} \sum_i p_i \rr_i(\tilde{\mu})
\end{equation}
where the minimum is taken over $\tilde{\mu} \in  \PPinvtheta(\Gamma_0)$. By the ergodic decomposition theorem for random dynamical systems (see Lemma 6.19 in \cite{crauel-02}), the minimum of the right hand side of \eqref{eq:minmax}
is attained at an ergodic probability measure with support in $\Gamma_0$. Thus, the equivalence of the conditions is established. 

Finally, we show the equivalence of condition (b) and (c). As a direct consequence of assertion (i) of Proposition \ref{prop:equality-lyap}, condition (c) implies (b). To prove the other direction, let $\tilde{\Omega} \subset \Omega$ be defined as in the proof of Proposition \ref{prop_walters}. Choose $(\omega',x') \in \tilde{\Omega} \times \bS_0$. By the same arguments as above, there exist $T>0$, $\gamma \in \pi_0^{-1}(\Phi^T(\omega',x')) \subset \Gamma_0$ and $\tilde{\mu} \in \PPinvtheta(\Gamma_0)$ such that
\begin{eqnarray*}
\rr_i(\tilde{\mu})&=&\int_{\Gamma} \bzeta(\eta) d\tilde{\mu}(\eta) \\
&=& \lim_{k \to \infty}\frac{1}{t_k}\sum_{s=0}^{t_k-1}\bzeta(\Theta^s(\gamma))\\
&\le&\rr_i(\gamma).
\end{eqnarray*}
Assertion (i) of Proposition \ref{prop:equality-lyap} implies that $\rr_i(\gamma)=\rr_i(\Phi^T(\omega',x'))$. Since $\Phi^T(\omega',x')$ is on the same trajectory that $(\omega',x')$, $\rr_i(\tilde{\mu}) \le \rr_i(\Phi^T(\omega',x'))= \rr_i(\omega',x')$. Writing $\tilde{\mu}$ as a convex combination of ergodic probability measures, condition (b) implies $\sum_i p_i \rr_i(\omega',x')>0$. $\qed$
\paragraph{{\bfseries Acknowledgements}}  GR was supported by the Swiss National Science Foundation Grant 137273 and a start-up grant to SJS from the College of Biological Sciences, University of California, Davis. SJS was supported in part by U.S. National Science Foundation Grants EF-0928987 and DMS-1022639.

\bibliographystyle{plainnat}

\bibliography{roth-schreiber-revised}

\end{document}